\documentclass[12pt,a4paper]{article}

\usepackage{algorithm, amsfonts, amsmath, amssymb, amsthm, array}
\usepackage{bbm, bm, bold-extra, booktabs}
\usepackage{color,comment}
\usepackage{enumerate}
\usepackage[T1]{fontenc}
\usepackage[left = 2.5cm, right = 2.5cm, top = 3cm, bottom = 2.5cm]{geometry}
\usepackage{graphicx}
\usepackage{hyperref}
\usepackage{mathrsfs, mathtools, multicol, multirow}
\usepackage{natbib}
\usepackage{setspace}
\usepackage{tikz}

\newcommand{\vertiii}[1]{{\left\vert\kern-0.25ex\left\vert\kern-0.25ex\left\vert #1 \right\vert\kern-0.25ex\right\vert\kern-0.25ex\right\vert}}

\newtheorem{prop}{Proposition}[section]

\newtheorem{thm}{Theorem}[section]

\newtheorem{corollary}{Corollary}[section]

\newcommand{\R}{\mathbb{R}}

\newcommand{\modch}{\textcolor{red}}

\newtheoremstyle{mytheorem}
  {\topsep}
  {\topsep}
  {\itshape}
  {0pt}
  {\bfseries}
  {.}
  { }
  {\thmname{#1}.\thmnumber{#2}\thmnote{ (#3)}}
\theoremstyle{mytheorem}

\newtheorem{asm}{A}
\labelformat{asm}{A.\theasm}

\newtheorem{cond}{C}
\labelformat{cond}{C.\thecond}

\long\def\symbolfootnote[#1]#2{\begingroup\def\thefootnote{\fnsymbol{footnote}}\footnote[#1]{#2}\endgroup}
 
\title {\large{\textsc{Concentration Inequalities for Suprema of Empirical Processes with Dependent Data via Generic Chaining with Applications to Statistical Learning}}} 
\author{\normalsize{Chiara Amorino}$^{\dag,*}$ \and \normalsize{Christian Brownlees}$^{\dag}$ \and \normalsize{Ankita Ghosh}$^{\dag}$} 

\begin{document} 

\maketitle

\begin{abstract}

This paper develops a general concentration inequality for the suprema of empirical processes with dependent data.
The concentration inequality is obtained by combining generic chaining with a coupling-based strategy.
Our framework accommodates high-dimensional and heavy-tailed (sub-Weibull) data. 
We demonstrate the usefulness of our result by deriving non-asymptotic predictive performance guarantees for empirical risk minimization in regression problems with dependent data. 
In particular, we establish an oracle inequality for a broad class of nonlinear regression models and, as a special case, a single-layer neural network model. 
Our results show that empirical risk minimizaton with dependent data attains a prediction accuracy comparable to that in the i.i.d.~setting for a wide range of nonlinear regression models.



\end{abstract}

\symbolfootnote[0]{\\
\noindent
$^{\dag}$ Department of Economics and Business, Universitat Pompeu Fabra and Barcelona SE;\\
e-mail: \texttt{chiara.amorino@upf.edu}, \texttt{christian.brownlees@upf.edu}, \texttt{ankita.ghosh@upf.edu}.\\
$^*$ Corresponding author. \\ 
We have benefited from discussions with Gabor Lugosi.\\
Christian Brownlees acknowledges support from the Spanish Ministry of Science and Technology (Grant MTM2012-37195);
the Ayudas Fundaci\'on BBVA Proyectos de Investigación Cient\`ifica en Matemáticas 2021;
the Spanish Ministry of Economy and Competitiveness through the Severo Ochoa Programme for Centres of Excellence in R\&D (SEV-2011-0075).}

\doublespace
\clearpage

\section{Introduction}

Bounds for the suprema of stochastic processes have numerous applications in statistics, econometrics, and machine learning. 
A powerful and general technique used to obtain such bounds is generic chaining.
Classical chaining consists of bounding the supremum of a stochastic process by constructing a sequence of increasingly fine partitions of the index set and appropriately controlling the process’ increments across the different partitions.
Generic chaining refines classical chaining by optimizing over the admissible sequences of partitions, typically leading to sharper bounds.
This technique was pioneered by Michel Talagrand, who was awarded the Abel Prize in 2024 in part for this key contribution to the theory of stochastic processes.
\citet{Talagrand2005}, \citet{Vandervaart} and \citet{Boucheron} provide, among others, a comprehensive treatment of this topic.

A classic application of generic chaining consists in obtaining concentration inequalities for the suprema of empirical processes.
The majority of applications in the literature, however, rely on the assumption that the underlying data are independent and identically distributed.
This is not appealing for applications in econometrics, where it is often more realistic to assume that the data exhibit dependence. 
This paper establishes a novel general concentration inequality for suprema of empirical processes with dependent data.
We do so by combining the generic chaining argument \citep{Talagrand2005} with a coupling argument to deal with the dependence \citep{MerlevedePeligrad2002}.
We demonstrate the usefulness of our result by obtaining non-asymptotic predictive performance guarantees for empirical risk minimization in statistical learning problems.

We begin by introducing a general concentration result for the supremum of an empirical process with dependent data.
We consider a (possibly nonlinear) function that depends on a random vector and a parameter belonging to some parameter space.
We then study the empirical process indexed by the parameter which is given by the average of the functions over a sequence of dependent random vectors.
The dependence structure of the sequence is characterized using the notion of $\beta$-mixing \citep{Doukhan:1994}.
Our main theorem is based on two high-level assumptions: an \emph{increment} condition and a \emph{coupling} condition. 
These two conditions allow us to develop, respectively, the generic chaining and the coupling arguments required to establish the main claim of the theorem.
The increment condition states that the sub-Weibull quasi-norm of the difference of the function evaluated in two parameter values, for the same random vector, is bounded by their distance.
The coupling condition states that the expected supremum (over the parameter space) of the absolute difference of the function evaluated in two random vectors, for the same parameter, is bounded by the $L_r$-norm of their distance.
This condition enables the use of a coupling lemma \citep[Theorem 2.9]{MerlevedePeligrad2002}, which allows us to approximate the sequence of dependent random vectors with an i.i.d.~sequence of random vectors with the same marginal distribution.

Our main theorem establishes a general concentration inequality for the supremum of empirical processes with dependent data, extending classical i.i.d.~results. 
The bound on the supremum of the empirical process depends on the so called Talagrand's functional, which captures the complexity of the parameter space, and on a coupling correction term accounting for the approximation error introduced when replacing the sequence of dependent random vectors with an independent copy.
The bounds is governed by a key quantity that we refer to as the \emph{effective sample size}. When observations are dependent, each additional observation provides less incremental information compared to the i.i.d.~case, and the effective sample size quantifies this loss of information due to dependence.

We apply our concentration result to study the properties of empirical risk minimization. 
Empirical risk minimization is a classic principle in statistical learning theory to choose a prediction rule for forecasting.
It consists of choosing the prediction rule that minimizes the average loss over the observed data, which is called the empirical risk.
A central problem in statistical learning theory is to understand the predictive performance of the empirical risk minimizer (ERM). 
Using our results, we derive predictive performance guarantees for the ERM for nonlinear regression with dependent data.
In particular, we establishes a non-asymptotic oracle inequality for the ERM under mild conditions on the regression model. 
The result implies that the predictive performance of the ERM approaches the best attainable performance at a rate that matches the so-called ``classical'' convergence rate of empirical risk minimization \citep[Ch.~12]{Devroye:Gyorfi:Lugosi:1996} once the sample size is replaced by the effective sample size. As a special illustration of the general framework, we obtain predictive performance guarantees for a single-layer neural network model. 
Overall, our results show that empirical risk minimzaton with dependent data attains a prediction accuracy comparable to that in the i.i.d.~setting for a wide range of nonlinear regression models.

This paper is related to different strands of the literature.
First it is related to the literature on generic chaining. In addition to the works we have already cited, additional important references on chaining and generic chaining include \citet{pollard1984convergence}, \citet{vandegeer2000} and \citet{Kosorok}. 
Introductory exposition on chaining and generic chaining is provided by \citet{Wainwright_2019} and \citet{Vershynin}, among others.
Second, this work is related to the literature on empirical risk minimization with dependent data.
Contributions in this literature include \citet{Jiang:Tanner:2010}, \citet{Brownlees:Gudmundsson:2021} and \citet{Brownlees:LlorensTerrazas:2021}.

The rest of the paper is outlined as follows.
Section \ref{sec:main} introduces the basic framework, the assumptions and the main theorem of this paper.
Section \ref{sec:learning} applies the main theorem in the context of statistical learning to obtain non-asymptotic prediction performance guarantees for empirical risk minimization for a fairly large class of nonlinear regression models and, as a special case, single-layer neural network.
Section \ref{sec:mainproof} outlines the proof of the main theorem.
Concluding remarks follow in Section \ref{sec:end}.
Additional proofs and results are collected in Appendix \ref{sec:proofs}.

\section{Basic Framework, Assumptions and Main Result}\label{sec:main}

Let \( \{ g_{\bm \theta} : \bm \theta \in \Theta \} \) be a class of real-valued functions defined on \( \mathcal Z \subset \mathbb R^d \) indexed by $\bm \theta \in \Theta$.
Let $\{ \bm Z_t \}$ be a dependent sequence of random vectors where $\bm Z_t$  takes values in $\mathcal Z$ for each $t$.
Our main objective consists in controlling the supremum of the empirical process associated with the average of the functions $g_{\bm \theta}$ based on the a sequence $\{ \bm Z_1, \ldots, \bm Z_T \}$, that is
\[
    \sup_{\bm \theta \in \Theta} \left| {1\over T} \sum_{t=1}^T g_{\bm \theta}(\bm Z_t) - \mathbb E(g_{\bm \theta}(\bm Z_t)) \right| ~.
\]
In what follows we refer to $T$ as the sample size.
Such a problem arises frequently in statistics, econometrics and machine learning. 
In the following section, we will show how controlling the supremum of the empirical process is key to obtain prediction performance guarantees in statistical learning problems.

Our concentration result relies on two high level assumptions that we present below.
Before stating the first of these two assumptions we need to introduce the notion of a sub-Weibull random variable of order $\alpha$, for some $\alpha > 0$ \citep{WongTewari:2020}.  
Let $\psi_\alpha(x) = \exp(x^\alpha) - 1$ for some $\alpha>0$ and define the quasi-norm of a random variable $X$ as
\begin{equation}{\label{eq: def subWei}}
 \| X \|_{\psi_\alpha} = \inf \left\{ c > 0 : \mathbb E\left( \psi_\alpha\left(\frac{|X|}{c}\right) \right) \leq 1 \right\} ~.   
\end{equation}
We refer to $\| \cdot \|_{\psi_\alpha}$ as the sub-Weibull($\alpha$) quasi-norm, and we say that a random variable $X$ is sub-Weibull($\alpha$) if $\| X \|_{\psi_\alpha} < \infty$.
We recall that the special cases $\alpha = 2$ and $\alpha = 1$ correspond, respectively, to the familiar notions of sub-Gaussian and sub-exponential random variables. 
Additional details and properties of sub-Weibull($\alpha$) random variables are provided in Appendix~\ref{App:sub-Wei}.

\begin{asm}[Increment condition]\label{asm:increments_and_tail}
There exists a distance $d_\Theta : \Theta \times \Theta \rightarrow \mathbb R_+$ and a positive constant $C_{\Theta}$ 
such that for all $t=1,\ldots,T$ we have that 
$(i)$ 
for any $\bm \theta_1, \bm \theta_2 \in \Theta$ it holds that 
\[ 
	\| g_{\bm\theta_1}(\bm Z_t) - g_{\bm\theta_2}(\bm Z_t) - \mathbb E ( g_{\bm\theta_1}(\bm Z_t) - g_{\bm\theta_2}(\bm Z_t) ) \|_{\psi_\alpha} \leq C_{\Theta} d_{\Theta}( \bm \theta_1 , \bm \theta_2 ) ~,
\]
and $(ii)$ for some ${\bm \theta}_0 \in \Theta$ it holds that $ \| g_{\bm\theta_0}(\bm Z_t) - \mathbb E g_{\bm\theta_0}(\bm Z_t) \|_{\psi_\alpha} \leq C_{\Theta} $.
\end{asm}

\ref{asm:increments_and_tail} implies that the increments of the empirical process exhibit sub-Weibull-type behaviour. 
This is a standard type of condition required to develop the chaining argument.
We remark that we state \ref{asm:increments_and_tail} for demeaned random variables for convenience.
It follows from the basic properties of sub-Weibull random variables that if $X$ is sub-Weibull of order $\alpha$ then $X - \mathbb E(X)$ is also sub-Weibull of order $\alpha$ (Proposition \ref{prop:subweibull:centering}).

\begin{asm}[Coupling condition]\label{asm:coupling}
There exists a distance $d_{\mathcal Z} : \mathcal Z \times \mathcal Z \rightarrow \mathbb R_+$, an $r\geq 1$ and a positive constant $C_{\mathcal Z}$ 
such that
$(i)$ $(\mathcal Z,d_{\mathcal Z})$ is a Polish space,
$(ii)$ for all $t=1,\ldots,T$ we have that it holds that
\[
	\mathbb E\left( \sup_{\bm \theta \in \Theta} |g_{\bm\theta}(\bm Z_t) - g_{\bm\theta}(\bm Z_t^*)| \right) 
	\leq \| d_{\mathcal Z}( \bm Z_t , \bm Z_t^* ) \|_{L_r} ~,
\]
where $\bm Z_t^*$ is random vector with the same marginal distribution as $\bm Z_t$ and
$(iii)$ for all $t=1,\ldots,T$ for some $\bm z \in \mathcal Z$ and some $s>0$ it holds that $\| d_{\mathcal Z}( \bm Z_t , \bm z ) \|_{L_{r+s}} < C_{\mathcal Z}$.
\end{asm}
\ref{asm:coupling} implies that the expected absolute difference between the empirical processes associated with two copies of the sequence $\{ \bm Z_t \}$ 
can be bounded by the average $L_r$-norm of the distance between the random vectors in the two sequences.
We also remark that the requirement that $(\mathcal Z,d_{\mathcal Z})$ is a Polish space is a technical condition required to apply a coupling result and that it is typically straightforward to verify.
This is a key condition required to develop the coupling argument.

Before stating our main concentration result, we introduce two key concepts: Talagrand's functional and the absolute regularity coefficients.

Talagrand's functional is a measure of complexity of a class of functions \citep{Talagrand2005}.
We say that a sequence of partition $\{ \mathcal A_k \}_{k\geq 0}$ of $\Theta$ is admissible if the sequence is increasing\footnote{An increasing sequence of partitions means that every set of $\mathcal A_{k+1}$ is included in a set of $\mathcal A_k$.} and it is such that $| \mathcal A_k | \leq 2^{2^k}$ for $k=0,1,\ldots$.
For any $\bm \theta \in \Theta$, denote by $A_k(\bm \theta)$ the unique element of $\mathcal A_k$ that contains $\bm \theta$.
Let $\Delta(A)$ denote the diameter of the set $A \subset \Theta$ associated with the distance $d_\Theta$. 
Finally, for $\alpha>0$ Talagrand's functional $\gamma_\alpha$ is defined as
\[
	\gamma_\alpha(\Theta) = \inf_{\mathcal A_k} \sup_{\bm \theta \in \Theta} \sum_{k\geq 0} 2^{k/\alpha} \Delta( A_k(\bm \theta) ) ~,
\]
where the infimum is taken over all admissible sequences.
It follows from standard arguments (Proposition \ref{prop:gammafunctional:bound}) that 
\[
	\gamma_\alpha(\Theta) \leq \log (2)^{1/\alpha} \left(1 - \frac{1}{2^{1/\alpha}}\right) \int_{0}^{\Delta(\Theta)} \big(\log \mathcal{N}(\Theta,\varepsilon)\big)^{1/\alpha} \, d\varepsilon ~,
\]
where $\mathcal{N}(\Theta, \varepsilon)$ denotes the covering number of $\Theta$ at scale $\varepsilon > 0$.

The absolute regularity coefficients, also known as $\beta$-mixing coefficients, measure the degree of dependence among the coordinates of the process $\{ \bm Z_t \}$ \citep{Doukhan:1994}.
Let $\mathcal{F}_{-\infty}^t$ and $\mathcal{F}_{t+l}^{\infty}$ be the $\sigma$-algebras generated by $\lbrace \bm Z_s: -\infty \leq s \leq t\rbrace$ and $\lbrace \bm Z_s: t + l \leq s \leq \infty\rbrace$ respectively. The $\beta$-mixing coefficient of order $l$, for $l\geq 0$, is defined as
\begin{equation*}
	\beta(l) = \sup_t \left\{ \sup_{\mathcal U, \mathcal V} {1\over 2} \sum_{i=1}^I \sum_{j=1}^J \left| {\mathbb P \left(U_i \cap V_j \right) - \mathbb P \left(U_i\right) \mathbb P \left(V_j\right) } \right| \right\}~,
\end{equation*}
where the inner supremum in the definition is taken over all pairs of partitions $\mathcal U = \{ U_1, \ldots, U_I \}$ and $\mathcal V = \{ V_1, \ldots, V_J \}$ of the sample space such that 
$U_i \in \mathcal{F}_{-\infty}^t$ and $V_j \in \mathcal{F}_{t+l}^{\infty}$ for all $i,j$.

Finally, we can state our main theorem.

\begin{thm}[Concentration]\label{thm:maintheorem}
Suppose \ref{asm:increments_and_tail} and \ref{asm:coupling} are satisfied.

Then, for any $n \in \{ 1 , \ldots , T \}$, any $\varepsilon_1\geq 2$ and any $\varepsilon_2>0$ 
\begin{align*}
	& \sup_{\bm\theta\in\Theta}\left|{1 \over T} \sum_{t=1}^T g_{\bm\theta}(\bm Z_t) - \mathbb E(g_{\bm\theta}(\bm Z_t)) \right| 
	\geq C_\alpha C_{\Theta} \left( { 1 + \gamma_2(\Theta)\over \sqrt{n}}  \varepsilon_1^{1/2} +{1+ \gamma_\alpha(\Theta)\over n^{1/\alpha \vee 1}}  \varepsilon_1^{1/\alpha} \right) + C_{\mathcal Z} \varepsilon_2 
\end{align*}
holds at most with probability
\[
	5 {T \over n} \exp( - \varepsilon_1 ) + 8 {T \over n} \beta^{s/(r(r+s))}\left( \left\lfloor{T \over n+1}\right\rfloor \right) {1 \over \varepsilon_2} ~,
\]
where $C_\alpha$ is a positive constant that depends on $\alpha$.
\end{thm}

A few remarks on Theorem \ref{thm:maintheorem} are in order.
To simplify the discussion it is useful to introduce a special version of the theorem.
Our result implies that for any $n \in \{ 1 , \ldots , T \}$ and any $\varepsilon\geq 2$ the inequality
\begin{align}
	& \sup_{\bm\theta\in\Theta}\left|{1 \over T} \sum_{t=1}^T g_{\bm\theta}(\bm Z_t) - \mathbb E(g_{\bm\theta}(\bm Z_t)) \right| \nonumber \\
	& \quad \geq C_\alpha C_{\Theta} \left( { 1 + \gamma_2(\Theta)\over \sqrt{n}}  \varepsilon^{1/2} +{1+ \gamma_\alpha(\Theta)\over n^{1/\alpha \vee 1}}  \varepsilon^{1/\alpha} \right) 
	+ C_{\mathcal Z} \beta^{s/(r(r+s))}\left( \left\lfloor{T \over n+1}\right\rfloor \right) \exp( \varepsilon ) \label{eqn:simple}
\end{align}
holds at most with probability \(13 (T/n)\exp( - \varepsilon ) \), where $C_\alpha$ is the same constant that appears in the statement of the theorem. This is obtained by choosing $\varepsilon_1 = \varepsilon$ and $\varepsilon_2 = \beta^{s/(r(r+s))}\left( \left\lfloor{T \over n+1}\right\rfloor \right) \exp( \varepsilon)$.

First, it is important to emphasize that the bound on the supremum of the empirical process is controlled by the variable $n$, which may be interpreted as the effective sample size. Intuitively, when observations are dependent the incremental information provided by an additional observation is in some sense smaller in comparison to the i.i.d.~case and 
the variable $n$ captures the loss of information due to dependence. 

Second, the supremum of the empirical process is bounded by three terms. 
The first two terms depend, respectively, on the Talagrand's functionals $\gamma_2(\Theta)$ and $\gamma_\alpha(\Theta)$, which capture the complexity of the parameter space $\Theta$.
When \( \alpha \geq 2 \) the first term dominates and we recover the classic sub-Gaussian concentration rate, with the effective sample size \( n \) playing the role typically held by the (actual) sample size \( T \) in the i.i.d.~setting.
On the contrary, when $\alpha<2$, the second term dominates, leading to a slower concentration rate, which is still controlled by the effective sample size $n$.
The third term depends on the $\beta$-mixing coefficients and may be interpreted as a correction term arising from the fact that the sequence of random vectors is dependent rather than independent.
It is important to highlight that the choice of effective sample size $n$ entails a trade-off (assuming that the $\beta$-mixing coefficients are decaying).
The first two terms that depend on Talagrand's functional are small when the effective sample size is large.
On the contrary the third term that depends on the $\beta$-mixing coefficients is small when the effective sample size is also small.c

Third, the probability bound of the inequality is the classic exponential-type bound that is typically associated with analogous concentration results for i.i.d.~data multiplied by the factor $T/n$.
The factor $T/n$ may be interpreted as a correction factor capturing the error arising from the fact that the sequence of random vectors is dependent rather than independent.

Fourth, the dimensionality of the parameter space affects the inequality through Talagrand's functional. In general, the larger the dimensionality of the parameter space the larger is Talagrand's functional.
The dimensionality of the data affects the inequality through the constant $C_Z$.
In what follows, we shall see how these constants simplify in the context of specific applications of our result.

Fifth, it is important to highlight that the generic chaining proof requires 
the empirical process to be \emph{separable}, in the sense of the definition of \citep[p.~305]{Boucheron}. 
In line with many authors, we assume throughout that these requirements are satisfied.

Last, we conclude with a few minor remarks on a number of additional aspects of theorem.
We note that the theorem holds for any $T$, unlike results stated in the literature, which are often stated to hold for an unspecified and sufficiently large $T$ \citep{Jiang:Tanner:2010,Brownlees:Gudmundsson:2021,Brownlees:LlorensTerrazas:2021}.
All the constants in the theorem can be recovered from the proofs in the appendix of the paper.
We do not provide explicit expressions in the text to avoid burdening exposition.
The theorem does not assume any specific rate of decay of the $\beta$-mixing coefficients. 
However, meaningful applications of the theorem require that the $\beta$-mixing coefficient decay at suitable rate.
Finally, applications of the theorem also require to set appropriately some of the variables in the theorem. 
We shall illustrate these choices in the application to statistical learning problems in the next section.



\section{Application to Statistical Learning Theory}\label{sec:learning}

Consider the stationary time series $\{ (Y_t,\bm X_t')' \}$ where $Y_t$ takes values in $\mathcal Y \subset \mathbb R$ and $\bm X_t$ takes values in $\mathcal X \subset \mathbb R^d$, with both $\mathcal Y$ and $\mathcal X$ assumed to be closed sets.
We are interested in forecasting the prediction target $Y_t$ on the basis of the vector of predictors $\bm X_t$.
The forecasts for the prediction target $Y_t$ are obtained from the class of prediction rules \( f_{\bm \theta} : \mathcal X \rightarrow \mathcal Y \) indexed by $\bm \theta \in \Theta$.
The square loss is used to measure prediction accuracy 
\[
	L( Y_t, f_{\bm \theta}(\bm X_t) ) := (  Y_t - f_{\bm \theta}(\bm X_t) )^2 ~.
\]
A standard problem in statistical learning consists is devising an algorithm to choose an accurate prediction rule $f_{\bm \theta}$ on the basis of a sample of observations
$\mathcal D = \{ (Y_1,\bm X_1')', \ldots, (Y_T,\bm X_T')' \}$.
One of the natural principles used to tackle this challenge is empirical risk minimization. This principle consists in choosing the $\bm \theta$ that minimizes the empirical risk, that is 
\[
	\hat{\bm \theta} \in \arg \min_{\bm \theta \in \Theta} R_T(\bm \theta) \text{ where } R_T(\bm \theta) := {1\over T}\sum_{t=1}^T (  Y_t - f_{\bm \theta}(\bm X_t) )^2  ~.
\]
If more than one $\bm \theta$ achieves the minimum we may pick one arbitrarily. We call $\hat {\bm \theta}$ the empirical risk minimizer (ERM).

The accuracy of the ERM is measured by its conditional risk defined as
\begin{equation}\label{eqn:risk:erm}
	R( \hat{ \bm \theta} )  
	:= \mathbb E( ( Y - f_{\hat{\bm \theta}}(\bm X) )^2 | \hat{ \bm \theta} = \hat{ \bm\theta} ( \mathcal D ) ) ~,
\end{equation}
where $(Y,\bm X')'$ denotes a draw form the stationary distribution of the time series $\{ (Y_t,\bm X_t')' \}$, and is assumed to be independent of the sample $\mathcal D$.
The performance measure in \eqref{eqn:risk:erm} can be interpreted as the risk of the ERM obtained from the ``training sample'' $\mathcal D$ over the ``validation observation'' $(Y,\bm X')'$.
This performance measure allows us to keep our analysis close to the bulk of contributions in the learning theory literature (which typically focus on the analysis of i.i.d.~data) and facilitates comparisons.
We remark that \citet{Brownlees:Gudmundsson:2021} and \citet{Brownlees:LlorensTerrazas:2021} consider 
alternative performance measures such as the conditional out-of-sample average risk of the ERM, 
which has a more attractive interpretation for time series applications. 
It turns out that these alternative measures lead to essentially the same theoretical analysis, 
at the expense of introducing additional notation.
Therefore, we focus on the performance measure defined in \eqref{eqn:risk:erm} for clarity.

A classic objective of statistical learning theory is to obtain a bound on the performance of the ERM relative to the optimal risk that can be achieved within the given class of prediction rules.
Define \( R(\bm \theta) := \mathbb E\left( (  Y_t - f_{\bm \theta}(\bm X_t) )^2\right) \).
Our aim is to find a pair $(B_T(\Theta),\delta_T)$ such that 
\begin{equation}\label{eqn:risk:regret}
	R( \hat{ \bm \theta } ) \leq \inf_{\bm \theta \in \Theta} R(\bm \theta) + B_T(\Theta) 
\end{equation}
holds at least with probability $1-\delta_T$ for all (sufficiently large) $T$.
In general, inequalities such as \eqref{eqn:risk:regret} provide non-asymptotic guarantees on the performance of the ERM.
Additionally, when we have that $B_T(\Theta) \rightarrow 0$ and $\delta_T \rightarrow 0$ as $T\rightarrow \infty$ the inequality in \eqref{eqn:risk:regret} is referred to as an \emph{oracle inequality},
meaning that that the ERM asymptotically performs as well as the best prediction rule in the class (when it exists).

Theorem \ref{thm:maintheorem} can be used to obtain performance bounds for empirical risk minimization.
We begin by recalling the basic inequality \citep[Lemma 8.2]{Devroye:Gyorfi:Lugosi:1996} stating that 
\[
	| R( \hat{\bm \theta} ) - \inf_{\bm \theta \in \theta} R(\bm \theta) | \leq 2 \sup_{\bm \theta \in \Theta} \left|  R_T(\bm \theta) - R(\bm \theta ) \right| ~.
\]

Let $\bm Z_t = (Y_t,\bm X_t')'$ that takes values in $\mathcal Z = \mathcal Y \times \mathcal X$
and define $g_{\bm \theta}(\bm Z_t) = (  Y_t - f(\bm X_t,\bm \theta) )^2$. Then, we have
\[
	\sup_{\bm \theta \in \Theta} \left|  R_T(\bm \theta) - R(\bm \theta ) \right|
	= \sup_{\bm\theta\in\Theta}\left|{1 \over T} \sum_{t=1}^T g_{\bm \theta}(\bm Z_t) - \mathbb E(g_{\bm \theta}(\bm Z_t)) \right| ~.
\]
Thus an application of Theorem \ref{thm:maintheorem} leads to the result of interest.

In order to apply Theorem \ref{thm:maintheorem}, we assume that a number of high-level conditions hold.
These conditions play the same role as \ref{asm:increments_and_tail} and \ref{asm:coupling} in the previous section, but are reformulated to fit the present framework.
In Section \ref{sec:neuralnet} we verify that these conditions are satisfied, for example, by a single-layer neural network model.

\begin{cond}\label{cond:c1}
There exists a distance $d_\Theta : \Theta \times \Theta \rightarrow \mathbb R_+$ and positive constant $C_1$ such that 
$(i)$ for any $\bm \theta_1$, $\bm \theta_2 $ in $\Theta$ it holds that 
\[
	\| f_{\bm \theta_1}(\bm X_t) - f_{\bm \theta_2}(\bm X_t) \|_{\psi_2} \leq C_1 d_\Theta( \bm \theta_1 , \bm \theta_2 ) ~,
\]
and $(ii)$ \( \| Y_t \|_{\psi_2} \leq C_1 \) and  \( \| \sup_{\theta \in \Theta} f_{\bm \theta}(\bm X_t) \|_{\psi_2} \leq C_1 \) \modch{.}
\end{cond}

Notice that in condition \ref{cond:c1} we have that the bound on the sub-Gaussian norms of 
$\sup_{\theta \in \Theta} f_{\bm \theta}(\bm X_t)$ and $f_{\bm \theta_1}(\bm X_t) - f_{\bm \theta_2}(\bm X_t)$ do not depend on dimension of $\bm X_t$.\footnote{This is satisfied, for example, when $f_{\bm \theta}(\bm X_t)=\bm X_t'\bm \theta$ and $\bm X_t$ is a sub-Gaussian vector. In this case for condition \ref{cond:c1}.$(i)$ we have that 
\[
	\| \bm X_t'(\bm \theta_1-\bm \theta_2) \|_{\psi_2} = \left\| \bm X_t'{(\bm \theta_1 - \bm \theta_2) \over \| \bm \theta_1 - \bm \theta_2 \|_2} \| \bm \theta_1 - \bm \theta_2 \|_2 
	\right\|_{\psi_2} \leq \| \bm X_t \|_{\psi_2} \| \bm \theta_1 - \bm \theta_2 \|_2 ~.
\]
}

\begin{cond}\label{cond:c2}
There exists a distance $d_{\mathcal X} : \mathcal X \times \mathcal X \rightarrow \mathbb R_+$ and positive constant $C_2$ such that 
$(i)$ $(\mathcal X,d_{\mathcal X})$ is a Polish space, 
$(ii)$ for any $\bm x_1$, $\bm x_2$ in $\mathcal X$ and any $\bm \theta$ in $\Theta$ it holds that
\begin{equation*}
	| f_{\bm \theta}(\bm x_1) - f_{\bm \theta}(\bm x_2) | \leq d_{\mathcal X}(\bm x_1 , \bm x_2 ) ~,
\end{equation*}
and $(iii)$ for some $\bm x$ in $\mathcal X$ it holds that $\| d_{\mathcal X}( \bm X_t , \bm x ) \|_{\psi_2} < C_2$.
\end{cond}

\begin{prop}\label{corollary:erm:reg}
	Suppose that \ref{cond:c1}--\ref{cond:c2} hold. 
	Suppose $\beta(l) \leq l^{-\zeta}$ for some $\zeta > 4$ and choose $n = \lceil T^\eta \rceil$ where $\eta=(\zeta-4)/(\zeta+2)$.

	Then, for any $T \geq 2$, 
	\begin{align*}
		R( \hat{\bm \theta} ) < \inf_{\bm \theta \in \theta} R(\bm \theta) 
		+ 
		C \left( \gamma_2(\Theta) \sqrt{ \log(n) \over n } + \gamma_1(\Theta) {\log(n)\over n} + {1\over \sqrt{n}} \right) ~,
	\end{align*}
	holds at least with probability \( 1-13 n^{-1}  \), where $C$ is a postive constant.
\end{prop}

A number of remarks on the proposition are in order.
First, the proposition implies that in our framework the ERM is consistent for prediction in the sense that $| R( \hat{\bm \theta} ) - \inf_{\bm \theta \in \theta} R(\bm \theta) | \stackrel{p}{\rightarrow} 0$.
Second, it is insightful to provide a simplified expression for the main claim of the proposition.
When $T$ is sufficiently large and assuming that the dimensionality of the parameter space and of the data is fixed we have that there is a positive constant $C$ such that 
\[
	R( \hat{\bm \theta} ) \leq \inf_{\bm \theta \in \theta} R(\bm \theta) + C \sqrt{ \log(n) \over n }
\]
holds at least with probability \( 1-13 n^{-1} \).
We recall that the rate of convergence $\sqrt{\log(n)/n}$ is typically referred to as the classical rate of convergence of empirical risk minimization in the learning literature with i.i.d.~data 
\citep[Ch.~12]{Devroye:Gyorfi:Lugosi:1996}. 
Thus, our results recovers the classical rate of converge once we replace the sample size $T$ with the effective sample size $n$.
We highlight that the proposition relies on fairly weak conditions on the sequence of mixing coefficients. 
In particular, it requires (a sufficiently fast rate of) polynomial decay as opposed to several contributions in the literature which typically assume geometric decay
\citep{Jiang:Tanner:2010,Brownlees:Gudmundsson:2021,Brownlees:LlorensTerrazas:2021}.
It is worth noting that the faster the rate of decay of the mixing coefficients (as captured by a larger value of $\zeta$), the smaller the discrepancy between $n$ and $T$ (as reflected by a value of $\eta$ closer to unity).

\subsection{Single-layer Neural Network}\label{sec:neuralnet}

It is instructive to apply Proposition~\ref{corollary:erm:reg} to a specific class of regression models in order to illustrate more concretely the implications of our results. 
In this section, we derive learning rates for a class of neural network models, specifically the single-layer perceptron for regression \citep[Ch.~11]{Hastieetal:2009}. 
We note that neural network models are typically trained using back-propagation algorithms rather than empirical risk minimization. 
Nevertheless, analyzing ERM remains valuable, as it offers theoretical benchmarks for understanding the predictive performance that can be expected to be achieved for this class of models.

The single-layer perceptron for regression may be defined as follows.
We start by defining a set of $K$ derived predictors called hidden units $H_{k\,t}$ for $k=1,\ldots,K$, which are nonlinear transformations of the original set of predictors. These are given by
\begin{equation}\label{eqn:nn:h}
	H_{k\,t}  = \sigma( \bm X_t' \bm w_k ) \text{ for } k=1,\ldots,K ~,
\end{equation}
where $\bm w_k$, $k=1,\ldots,K$, is a set of weight vectors and $\sigma : \mathbb R \rightarrow \mathbb R$ is the so-called activation function. 
Classic choices for $\sigma$ include the rectified linear unit (ReLU) function $\sigma(x)=\max\{0,x\}$ or the sigmoid function $\sigma(x)=1/(1+e^{-x})$.
We assume that $Y_t$ is subGaussian with subGaussian norm $\| \bm Y_t \|_{\psi_2} = \sigma_Y$, and that
$\bm X_t$ is subGaussian with subGaussian norm $\| \bm X_t \|_{\psi_2} = \sigma_X $.
Moreover, we assume that the activation function is sub-differentiable with a bounded first sub-derivative and that $\sigma(0)=0$.
Forecasts for the target variable $Y_t$ are then obtained by combining the hidden units 
\begin{equation}\label{eqn:nn:aggr}
	f_{\bm\theta\,t} = \sum_{k=1}^K \psi_k H_{k\,t} ~,
\end{equation}
where $\psi_1,\ldots,\psi_K$ are additional weights.
Putting together \eqref{eqn:nn:h} and \eqref{eqn:nn:aggr}, we get that the class of prediction rules in the single-layer perceptron for regression is given by 
\begin{align*}
	f_{\bm \theta}(\bm X_t)& = \sum_{k=1}^K \psi_k \sigma( \bm X_t' \bm w_k ) ~,
\end{align*}
with $\bm \theta = (\bm w_1',\ldots,\bm w_K',\psi_1,\ldots,\psi_K)' \in \mathbb R^p$ with $p=Kd+K$. We further assume that $\bm \theta$ belongs to the set $\Theta$ that is compact. 

The following corrollary specializes Proposition \ref{corollary:erm:reg} for the single-layer perceptron for regression.

\begin{corollary}\label{corollary:erm:neural_nets}

	Consider the class of prediction rules given by the single-layer perceptron for regression given by \eqref{eqn:nn:h} and \eqref{eqn:nn:aggr}.
	Suppose $\beta(l) \leq l^{-\zeta}$ for some $\zeta > 4$ and choose $n = \lceil T^\eta \rceil$ where $\eta=(\zeta-4)/(\zeta+2)$.

	Then, for any $T \geq 8$, 
	\begin{align*}
		R( \hat{\bm \theta} ) &\leq \inf_{\bm \theta \in \theta} R(\bm \theta) 
		+ 
		C \left( \sqrt{ d \log(n) \over n } + {d \log(n)\over n} + \sqrt{ \log(d) \over n } \right) ~,
	\end{align*}
	holds at least with probability \( 1-13 n^{-1}  \), 
	where $C$ is a positive constant. 
\end{corollary}

\section{Proof of Theorem \ref{thm:maintheorem}}\label{sec:mainproof}

In this section we detail the proof of Theorem \ref{thm:maintheorem}.
To simplify exposition throughout this section we use $g_{\bm\theta}(\bm Z_t)$ to denote $g_{\bm\theta}(\bm Z_t) - \mathbb E(g_{\bm\theta}(\bm Z_t))$.

First, we introduce a coupling result \citep[Theorem 2.9]{MerlevedePeligrad2002} that is key to the proof.

\begin{prop} \label{prop:coupling}
Let $\{ \bm X_i \}_{i=1}^n$ be a sequence of random vectors taking values in the set $\mathcal X \subset \mathbb R^d$ equipped with the metric $d_{\mathcal X}$ such that $(\mathcal X,d_{\mathcal X})$ is a Polish space.
Then, we can redefine $\{ \bm X_i \}_{i=1}^n$ onto a richer probability space together with a sequence of $\{ \bm X^*_i \}_{i=1}^n$ of independent random vectors such that for each $i \in \{1,\ldots,n\}$ we have that
\begin{enumerate}[(i)]
	\item $ \bm X^*_i $ has the same distribution as $\bm X_i$ and is independent of $\mathcal F_1^{i-1} = \sigma(\bm X_1,\ldots,\bm X_{i-1})$;
	\item if $\mathbb E( d_{\mathcal X}^r(\bm X_i,\bm x) ) < \infty $ for an $r>0$ and some $\bm x \in \mathcal X$ then it holds
	\[
		\mathbb E( d_{\mathcal X}^r(\bm X_i,\bm X_i^*) ) \leq 2^{r+2} \int_0^{\beta(\mathcal F_1^{i-1},\sigma(\bm X_i))} Q_{ d_{\mathcal X}^r(\bm X_i,\bm x) }(u) d u ~,
	\]
	where $Q_{W}(u) = \inf\{ q \geq 0 : \mathbb P( W > q ) \leq u \}$ denotes the ``upper tail'' quantile function associated with the nonnegative random variable $W$;
	\item if $\mathbb E( d_{\mathcal X}^{r+s}(\bm X_i,\bm x) ) < \infty $ for an $r>0$ and $s>0$ and some $\bm x \in \mathcal X$ then it holds
	\[
		\mathbb E( d_{\mathcal X}^r(\bm X_i,\bm X_i^*) ) \leq 
		2^{r+2} \beta^{s/(r+s)}(\mathcal F_1^{i-1},\sigma(\bm X_i)) ( \mathbb E( d_{\mathcal X}^{r+s}(\bm X_i^*,\bm x) ) )^{r/(r+s)}  ~.
	\]
\end{enumerate}
\end{prop}

Our proof strategy is built upon Proposition \ref{prop:coupling}.
Let $M$ be a natural number such that $T/(n+1) < M \leq T/n$.
Consider the extension of the sequence of vectors $\{\bm  Z_1, \ldots, \bm  Z_T \}$ given by $\{ \bm z, \bm Z_1, \ldots, \bm Z_T , \bm z, \bm z, \ldots \}$ where $\bm z$ denotes an arbitrary  
element in $\mathcal Z$ (which is deterministic). 
Define $\bm W_{i,j} = \bm Z_{i M +j}$ for $i\in \{0,\ldots,n\}$ and $j\in \{0,\ldots,M-1\}$.
For each $j \in \{0,\ldots,M-1\}$ consider the sequence $\{ \bm W^*_{0,j} , \ldots , \bm W^*_{n,j}\}$ constructed from $\{ \bm W_{0,j} , \ldots , \bm W_{n,j}\}$ using Proposition \ref{prop:coupling}.
Then we have  
\begin{align*}
	& \sup_{\theta \in \Theta} \left|{1\over T}\sum_{t=1}^T g_{\bm\theta}(\bm Z_t) \right| = 
	\sup_{\theta \in \Theta} {1\over T} \left|\sum_{t=1}^T g_{\bm\theta}(\bm Z_t) - \sum_{j=0}^{M-1} \sum_{i=0}^{n} g_{\bm\theta}(\bm W^*_{i,j}) - g_{\bm\theta}(\bm W^*_{i,j}) \right| \\
	& \quad 
	= \sup_{\theta \in \Theta} {1\over T} \left|\sum_{j=0}^{M-1} \sum_{i=0}^{n} g_{\bm\theta}(\bm W_{i,j}) - \sum_{j=0}^{M-1} \sum_{i=0}^{n} g_{\bm\theta}(\bm W^*_{i,j}) 
	+ \sum_{j=0}^{M-1} \sum_{i=0}^{n} g_{\bm\theta}(\bm W^*_{i,j}) \right| \\
	& \quad 
	\leq 
	\sup_{\theta \in \Theta}  {1\over Mn} \left| \sum_{j=0}^{M-1} \sum_{i=0}^{n} g_{\bm\theta}(\bm W^*_{i,j}) \right| 
	+ \sup_{\theta \in \Theta} {1\over Mn} \left|\sum_{j=0}^{M-1} \sum_{i=0}^{n} g_{\bm\theta}(\bm W_{i,j}) - \sum_{j=0}^{M-1} \sum_{i=0}^{n} g_{\bm\theta}(\bm W^*_{i,j}) \right| ~.
\end{align*}
Note that the second equality follows from the fact that $g_{\bm\theta}(\bm W_{i,j},\bm\theta)=0$ when $\bm W_{i,j}=\bm z$.
Furthermore, we have that
\begin{align*}
	& \sup_{\theta \in \Theta}  {1\over Mn} \left| \sum_{j=0}^{M-1} \sum_{i=0}^{n} g_{\bm\theta}(\bm W^*_{i,j}) \right| \\ 
	& \quad \leq {1\over Mn} \left| \sum_{j=0}^{M-1} \sum_{i=0}^{n} g_{\bm\theta_0}(\bm W^*_{i,j}) \right|
	+ \sup_{\theta \in \Theta}  {1\over Mn} \left| \sum_{j=0}^{M-1} \sum_{i=0}^{n} g_{\bm\theta}(\bm W^*_{i,j}) - g_{\bm\theta_0}(\bm W^*_{i,j})\right| ~,
\end{align*}
where $\bm \theta_0$ is defined \ref{asm:increments_and_tail}.$(ii)$.
Then, for any $\varepsilon', \varepsilon'_1, \varepsilon'_2, \varepsilon'_3 \geq 0$ such that $\varepsilon' = \varepsilon'_1 + \varepsilon'_2 + \varepsilon_3'$ we have that 
\begin{align}
	& \mathbb P\left(\sup_{\bm\theta\in\Theta}\left|{1 \over T} \sum_{t=1}^T g_{\bm\theta}(\bm Z_t) \right| \geq \varepsilon' \right) 	
	\leq \mathbb P\left( {1 \over Mn} \sum_{j=0}^{M-1}\left| \sum_{i=0}^{n} g_{\bm\theta_0}(\bm W^*_{i,j}) \right| \geq {\varepsilon'_1} \right) \nonumber \\
	& \quad + \mathbb P\left(\sup_{\bm\theta\in\Theta}{1 \over Mn} \sum_{j=0}^{M-1}\left| \sum_{i=0}^{n} g_{\bm\theta}(\bm W^*_{i,j})-  g_{\bm\theta_0}(\bm W^*_{i,j}) \right| \geq {\varepsilon'_2} \right) \nonumber \\
	& \quad + \mathbb P\left(\sup_{\bm\theta\in\Theta}{1 \over Mn}\sum_{j=0}^{M-1} \left| \sum_{i=0}^{n} g_{\bm\theta}(\bm W_{i,j}) - g_{\bm\theta}(\bm W^*_{i,j}) \right| \geq {\varepsilon'_3} \right) \nonumber \\
	& \quad \leq 
	{T \over n} \max_{0\leq j \leq M-1} \mathbb P\left( \left| {1 \over n}\sum_{i=0}^{n} g_{\bm\theta_0}(\bm W^*_{i,j}) \right| \geq {\varepsilon'_1} \right) \nonumber \\
	& \quad + {T \over n} \max_{0\leq j \leq M-1}  \mathbb P\left(\sup_{\bm\theta\in\Theta} \left| \sum_{i=0}^{n} g_{\bm\theta}(\bm W^*_{i,j})-  g_{\bm\theta_0}(\bm W^*_{i,j}) \right| \geq {\varepsilon'_2} \right) \nonumber \\
	& \quad + {T \over n} \max_{0\leq j \leq M-1}\mathbb P\left( \sup_{\bm\theta\in\Theta} \left| {1 \over n} \sum_{i=0}^{n} g_{\bm\theta}(\bm W_{i,j}) - g_{\bm\theta}(\bm W^*_{i,j}) \right| \geq {\varepsilon'_3} \right) ~. \label{eqn:basicdecomp} 
\end{align}
Our objective is to find appropriate bounds for the three terms in \eqref{eqn:basicdecomp}.

We begin with the first term in \eqref{eqn:basicdecomp}. 
We introduce a concentration result for sub-Weibull random variables that is based on \citep[Theorem 3.1]{KuchibhotlaChakrabortty2022} and results by \citet{Latala1997}.

\begin{prop}\label{prop:subweibull:conc}
Let $X_1,\ldots,X_n$ be independent zero-mean sub-Weibull($\alpha$) random variables of order $\alpha$ for some $\alpha > 0$ such that $\| X_i \|_{\psi_\alpha} < C_X$ for each $i = 1,\ldots, n$.

Then, for any $\varepsilon\geq 0$ it holds that
\[
	\mathbb P\left( \left| \sum_{i=1}^n X_i \right| \geq C'_\alpha C_X \sqrt{n} \sqrt{ \varepsilon } + C''_\alpha C_X  n^{(\alpha-1)/\alpha \vee 0} \varepsilon^{1/\alpha} \right) \leq e \exp( -\varepsilon ) ~,
\]
where $C'_\alpha$ and $C''_\alpha$ are constants that only depend on $\alpha$.
\end{prop}
We remark that the explict expressions for the constants $C'_\alpha$ and $C''_\alpha$ can be deduced in the proof of the proposition.

We note that \ref{asm:increments_and_tail}.$(ii)$ and Proposition \ref{prop:subweibull:conc} imply that for any $j \in \{ 0, \ldots, M-1 \}$ we have that
\begin{align}\label{eqn:part1}
	& \mathbb P\left( \left|{1 \over  n}\sum_{i=0}^{n} g_{\bm\theta_0}(\bm W^*_{i,j}) \right| 
	\geq {C'_\alpha C_{\Theta} \sqrt{(\varepsilon'_1)}\over \sqrt{n}}+ {C''_\alpha C_{\Theta} ( \varepsilon'_1)^{1\over \alpha}\over n^{1/\alpha\vee 1}} \right) 
	\leq e \exp\left( - \varepsilon'_1\right) ~. 
\end{align}
Notice that in this result we are using the fact that the random variable $g_{\bm\theta}(\bm W^*_{i,j})$ is degenerate at zero when $\bm W^*_{i,j} = \bm z$ and that in this case we have that 
$\| g_{\bm\theta}(\bm W^*_{i,j}) \|_{\psi_\alpha} < C_{\Theta}$.

We continue with the second term in \eqref{eqn:basicdecomp}. 
\ref{asm:increments_and_tail}.$(i)$ and Proposition \ref{prop:subweibull:conc} imply that for any $j \in \{ 0, \ldots, M-1 \}$, any $\bm \theta_1, \bm \theta_2 \in \Theta$ and any $\varepsilon\geq 0$ we have that
\begin{align*}
	& \mathbb P\left( \left|{1 \over  n}\sum_{i=0}^{n} g_{\bm\theta_1}(\bm W^*_{i,j}) - g_{\bm\theta_2}(\bm W^*_{i,j})  \right| \geq \left({C'_\alpha C_\Theta \sqrt\varepsilon\over \sqrt{n}}+ {C''_\alpha C_\Theta \varepsilon^{1\over \alpha}\over n^{1/\alpha\vee 1}} \right) d_{\Theta}(\bm\theta_1,\bm\theta_2)\right) 
	\leq e \exp\left( -\varepsilon \right) ~.
\end{align*}
In other words, the empirical process satisfies a sub-Weibull increment-type condition.
Such a property allows us to develop a generic chaining argument to control its supremum and, hence to control the second term in \eqref{eqn:basicdecomp}.

\begin{prop}
	[Generic Chaining] \label{prop:chaining:1} 
	Let $\{ X_{\bm \theta} \}_{\bm \theta \in \Theta}$ be a separable zero-mean stochastic process on a metric space $(\Theta,d_\Theta)$ 
	that satisfies for any $\bm \theta_1$, $\bm \theta_2$ in $\Theta$, some $a, b, \alpha>0$ and any $\varepsilon\geq0$
	\[
		\mathbb P\left( |X_{\bm \theta_1}-X_{\bm\theta_2}| \geq a d_\Theta(\bm\theta_1,\bm\theta_2) \sqrt{\varepsilon} + b d_\Theta(\bm\theta_1,\bm\theta_2) \varepsilon^{1/\alpha} \right) 
		\leq e \exp\left( - \varepsilon\right) ~.
	\]
	Then, for any $\bm \theta_0 \in \Theta$ and any $\varepsilon\geq 2$  the event
\begin{align*}
	\sup_{\bm\theta\in\Theta}|X_{\bm\theta}-X_{\bm\theta_0}| 
	\geq 8 a\gamma_2(\Theta)\sqrt\varepsilon+ 4^{(\alpha+1)/\alpha}b\gamma_\alpha(\Theta)\varepsilon^{1/\alpha}
\end{align*}
holds at most with probability $2 \exp(-\varepsilon)$.
\end{prop}

We outline here the basic strategy of the generic chaining proof.
We are interested in establishing a high-probability bound for 
\[
	\sup_{\bm\theta\in\Theta}|X_{\bm\theta}-X_{\bm\theta_0}| ~.
\]
To simplify exposition, here we assume that $\Theta$ is finite \citep[Ch.~2]{Talagrand2005}.\footnote{We remark that Proposition \ref{prop:chaining:1} does not rely on this assumption and allows $\Theta$ to be uncountable.}
We begin by constructing a sequence of subsets of $\Theta$ denoted by $\{ \Theta_k \}_{k\geq 0}^K$ such that $\bm \theta_0 \in \Theta_0$ and $\Theta = \Theta_K$.
The sequence of subsets is carefully constructed and may be interpreted as a sequence of progressively finer approximations of $\Theta$, in the sense that any $\bm \theta$ can be more accurately approximated by an element in $\Theta_k$ as $k$ increases.
Let $\pi_k(\bm \theta) = \arg \min_{ s \in \Theta_k} d_\Theta(s,\bm \theta) $ denote the closest element of the set $\Theta_k$ to $\bm \theta$.
Then, by constructing a telescoping sum and applying the triangle inequality we get that 
\[
	\sup_{\bm\theta\in\Theta}|X_{\bm\theta}-X_{\bm\theta_0}|\leq \sup_{\bm\theta\in\Theta}\sum_{k\geq 1}^K|X_{\pi_k(\bm\theta)}-X_{\pi_{k-1}(\bm\theta)}| ~.
\]
Next, for any \(\varepsilon \geq 0\), define the event \(\Omega(\varepsilon)\) as
\[
	\left\{ 
	\text{ for all } k \in \{ 1 , \ldots , K \}, \text{ for any } \bm\theta_1,\bm\theta_2\in \Theta_k,~ |X_{\bm \theta_1}-X_{\bm\theta_2}| \leq c_k(\varepsilon) d_\Theta(\bm\theta_1,\bm\theta_2) 
	~\right\} ~,
\]
where $c_k = a2^{(k+1)/2} \sqrt\varepsilon + b2^{(k+1)/\alpha}\varepsilon^{1/\alpha}$. 
It can be shown that, under the sub-Weibull increment condition, the event $\Omega^c(\varepsilon)$ is realized with probability at most $2\exp(-\varepsilon)$ for any $\varepsilon\geq 2$.
Then, assuming that the \(\Omega(\varepsilon)\) is realized we have that 
\begin{align*}
& \sup_{\bm\theta\in\Theta}|X_{\bm\theta}-X_{\bm\theta_0}|\leq \sup_{\bm\theta\in\Theta}\sum_{k\geq 1}^K|X_{\pi_k(\bm\theta)}-X_{\pi_{k-1}(\bm\theta)}| 
\leq\sup_{\bm\theta\in\Theta}\sum_{k\geq 1}^K c_k(\varepsilon) d_\Theta( \pi_k(\bm\theta), \pi_{k-1}(\bm\theta) ) ~.
\end{align*}
The final upper bound follows from straightforward computations by studying the properties of the summation in the last display.

Condition \ref{asm:increments_and_tail}.$(i)$, Proposition \ref{prop:subweibull:conc} and Proposition \ref{prop:chaining:1} imply that for any $j \in \{ 0, \ldots, M-1 \}$ and any $\varepsilon'_2 \geq 2$ we have that
\begin{align}\label{eqn:part2}
	& \mathbb P\left(\sup_{\bm\theta\in\Theta} \left| \sum_{i=0}^{n} g_{\bm\theta}(\bm W^*_{i,j})-  g_{\bm\theta_0}(\bm W^*_{i,j}) \right| 
	\geq 8 {C'_\alpha C_\Theta \gamma_2(\Theta)\sqrt{(\varepsilon'_2)}\over \sqrt{n}}+ 4^{(\alpha+1)/\alpha} {C''_\alpha C_\Theta \gamma_\alpha(\Theta) ( \varepsilon'_2)^{1\over \alpha}\over n^{1/\alpha\vee 1}} \right) \nonumber \\
	& \quad \leq 2 \exp\left(-\varepsilon_2'\right) ~.
\end{align}

We conclude with the third term in \eqref{eqn:basicdecomp}. 
\ref{asm:coupling} and Proposition \ref{prop:coupling} imply that for any $j \in \{ 0, \ldots, M-1 \}$, some ${\bm w} \in \mathcal Z$ and some $s>0$ we have that 
\begin{align}
	& \mathbb P\left( \sup_{\bm\theta\in\Theta} \left| {1 \over n} \sum_{i=0}^{n} g_{\bm\theta}(\bm W_{i,j}) - g_{\bm\theta}(\bm W^*_{i,j}) \right| \geq {\varepsilon'_3} \right) 
	\leq \mathbb P\left( {1 \over n} \sum_{i=0}^{n} \sup_{\bm\theta\in\Theta} \left|g_{\bm\theta}(\bm W_{i,j}) - g_{\bm\theta}(\bm W^*_{i,j}) \right| \geq {\varepsilon'_3} \right) \nonumber \\
	& \quad \leq {1 \over \varepsilon'_3} {1 \over n} \sum_{i=0}^{n} \mathbb E\left( \sup_{\bm\theta\in\Theta} \left|g_{\bm\theta}(\bm W_{i,j}) - g_{\bm\theta}(\bm W^*_{i,j}) \right| \right) \nonumber \\
	& \quad \leq {1 \over \varepsilon'_3} {1 \over n} \sum_{i=0}^{n} \| d_{\mathcal X}(\bm W_{i,j},\bm W^*_{i,j}) \|_{L_r}  
	\leq {1 \over \varepsilon'_3} {1 \over n} \sum_{i=0}^{n} ( 2^{r+2} \beta^{s/(r+s)}(M) ( \mathbb E( d_{\mathcal X}^{r+s}(\bm W_{i,j},\bm w) ) )^{r/(r+s)}   )^{1/r} \nonumber \\
	& \quad = {2^{(r+2)/r} C_{\mathcal Z} } \beta^{s/(r(r+s))}(M) {1 \over \varepsilon'_3} 
	\leq 8 C_{\mathcal Z} \beta^{s/(r(r+s))}\left(\left\lfloor{T \over n+1}\right\rfloor\right) {1 \over \varepsilon_3'} \label{eqn:part3} ~.
\end{align}

The claim of the theorem then follows from \eqref{eqn:basicdecomp}, \eqref{eqn:part1}, \eqref{eqn:part2} and \eqref{eqn:part3} after 
setting $\varepsilon_1' = \varepsilon_2'$ and redefining $\varepsilon_1 = \varepsilon_1'$ and $\varepsilon_2 = C_{\mathcal Z} \varepsilon_3'$. 


We conclude this section with an auxiliary proposition that provides an upper bound for Talagrand's functional in terms of a generalised version of Dudley's entropy integral.
This result allows to simplify the bounds of the empirical process implied by our main theorem in the applications.

\begin{prop}\label{prop:gammafunctional:bound}
	Consider the functional $\gamma_\alpha(\Theta)$ for some $\alpha>0$.

	Then, it holds that 
	\[
	\gamma_\alpha(\Theta) \leq (\log (2))^{1/\alpha} \left(1 - \frac{1}{2^{1/\alpha}}\right) \int_{0}^{\Delta(\Theta)} \big(\log \mathcal{N}(\Theta, \varepsilon)\big)^{1/\alpha} \, d\varepsilon ~,
	\]
	where $\mathcal{N}(\Theta, \varepsilon)$ is the covering number of the set $\Theta$ at scale $\varepsilon > 0$.
\end{prop}

\section{Conclusion}\label{sec:end}

This paper establishes a concentration inequality for the suprema of the empirical processes with dependent data.
The concentration inequality is established by developing an argument based on generic chaining combined with a coupling strategy. 
We apply our result to study the properties of statistical learning procedures. 
Specifically, we derive non-asymptotic predictive performance guarantees for empirical risk minimization for nonlinear regression.
We show that empirical risk minimization achieves the classical convergence rate that can be obtained in i.i.d.~setting after replacing the sample size with what we call 
in this work the effective sample size, a notion of sample size that reflects the loss of information due to the dependence with respect to the i.i.d.~case.
Our result encompasses a broad class of nonlinear regression models, including a single-layer neural network models,
and offers theoretical guarantees for widely used statistical learning procedures in dependent data environments.

\appendix

\section{Appendix}\label{sec:proofs}

\subsection{Proofs for Section \ref{sec:learning}}

\begin{proof}[Proof of Proposition \ref{corollary:erm:reg}]

We begin the proof by verifying that \ref{cond:c1} and \ref{cond:c2} imply that \ref{asm:increments_and_tail} and \ref{asm:coupling} are satisfied. We then apply Theorem \ref{thm:maintheorem}.

\paragraph{Verifying \ref{asm:increments_and_tail}.} 
For any $\bm \theta_1,\bm\theta_2\in\Theta$ it holds that 
\begin{align*}
	& (Y_t - f_{\bm \theta_{1}}(\bm X_t))^2 - (Y_t - f_{\bm \theta_{2}}(\bm X_t))^2 
	= (-2Y_t + f_{\bm \theta_{1}}(\bm X_t) + f_{\bm \theta_{2}}(\bm X_t))(f_{\bm \theta_{1}}(\bm X_t) - f_{\bm \theta_{2}}(\bm X_t)) \\
	& \quad \leq 2( |Y_t| + \sup_{\bm \theta} |f_{\bm\theta}(\bm X_t)| ) | f_{\bm \theta_{1}}(\bm X_t) - f_{\bm \theta_{2}}(\bm X_t) |~.
\end{align*}	
Then \ref{cond:c1} and basic properties of subGaussian random variables imply that $\| g_{\bm\theta_1}(\bm Z_t) -g_{\bm\theta_2}(\bm Z_t) \|_{\psi_1} \leq 4 C^2_1 d_\Theta(\bm \theta_1,\bm \theta_2)$.
For any $\bm\theta\in\Theta$ we have  
\[
	g_{\bm\theta}(\bm Z_t)=(Y_t-f_{\bm\theta}(\bm X_t))^2 \leq 2 Y_t^2 + 2 f_{\bm\theta}(\bm X_t)^2 \leq 2( Y_t^2 + \sup_{\bm \theta} f_{\bm\theta}(\bm X_t)^2 )~.
\]
Then \ref{cond:c1} and Proposition \ref{prop:subweibull:centering} imply that $ \|g_{\bm\theta}(\bm Z_t)-\mathbb E g_{\bm\theta}(\bm Z_t)\|_{\psi_1} \leq 16C_1^2$. 
Thus \ref{asm:increments_and_tail} is satisfied for $C_\Theta=16C_1^2$.

\paragraph{Verifying \ref{asm:coupling}.}
We note that 
\begin{align*}
	& \mathbb E\left( \sup_{\bm \theta \in \Theta} \left| g_{\bm\theta}(\bm Z_t) -g_{\bm\theta}(\bm Z_t^*) \right| \right) \\
	& \quad = \mathbb E\left( \sup_{\bm \theta \in \Theta} \left| (Y_t-f_{\bm \theta}(\bm X_t) - Y_t^*+f_{\bm\theta}(\bm X_t^*))^{2} + 2(Y_t^*-f_{\bm\theta}(\bm X_t^*))[(Y_t-f_{\bm\theta}(\bm X_t)) - (Y_t^*-f_{\bm\theta}(\bm X_t^*))] \right| \right) \\
	& \quad \leq \mathbb E\left( \sup_{\bm \theta \in \Theta} \left|Y_t-Y_t^* - f_{\bm \theta}(\bm X_t) +f_{\bm\theta}(\bm X_t^*)\right|^2 \right) \\ 
	& \quad + 2 \mathbb E\left( \sup_{\bm \theta \in \Theta} \left| (Y_t^*-f_{\bm\theta}(\bm X_t^*) ) ( Y_t-Y_t^* - f_{\bm \theta}(\bm X_t) +f_{\bm\theta}(\bm X_t^*) ) \right| \right) \\
	& \quad \leq \mathbb E\left( \sup_{\bm \theta \in \Theta} \left|Y_t-Y_t^* - f_{\bm \theta}(\bm X_t) +f_{\bm\theta}(\bm X_t^*)\right|^2 \right) \\ 
	& \quad + 2 \left( \mathbb E\left( \sup_{\bm \theta \in \Theta} \left| Y_t^*-f_{\bm\theta}(\bm X_t^*) \right|^2 \right) 
	\mathbb E\left( \sup_{\bm \theta \in \Theta} | Y_t-Y_t^* - f_{\bm \theta}(\bm X_t) +f_{\bm\theta}(\bm X_t^*) |^2 \right) \right)^{1/2} \\
	& \quad \leq 
	\left( \left\| \sup_{\bm \theta \in \Theta} \left|Y_t-Y_t^* - f_{\bm \theta}(\bm X_t) +f_{\bm\theta}(\bm X_t^*)\right| \right\|_{L_2} 
	+ 2 \left\| \sup_{\bm \theta \in \Theta} \left| Y_t^*-f_{\bm\theta}(\bm X_t^*) \right| \right\|_{L_2} \right) \\ 
	& \quad \times \left\| \sup_{\bm \theta \in \Theta} | Y_t-Y_t^* - f_{\bm \theta}(\bm X_t) +f_{\bm\theta}(\bm X_t^*) | \right\|_{L_2} \\
	& \quad \leq 
	4 \left\| \sup_{\bm \theta \in \Theta} \left| Y_t^*-f_{\bm\theta}(\bm X_t^*) \right| \right\|_{L_2} 
	\left\| | Y_t-Y_t^* | + d_{\mathcal X}(\bm X_t,\bm X_t^*) \right\|_{L_2} ~.
\end{align*}
Next we note that
\begin{align*}
	\left\| \sup_{\bm \theta \in \Theta} \left| Y_t^*-f_{\bm\theta}(\bm X_t^*) \right| \right\|_{L_2}
	\leq \| Y_t^* \|_{L_2} + \left\| \sup_{\bm \theta \in \Theta} f_{\bm\theta}(\bm X_t^*) \right\|_{L_2} \leq 12 C_1 ~.
\end{align*}
Thus we have that
\[
	\mathbb E\left( \sup_{\bm \theta \in \Theta} \left| g_{\bm\theta}(\bm Z_t) -g_{\bm\theta}(\bm Z_t^*) \right| \right) 
	\leq 12 C_1 \left\| | Y_t-Y_t^* | + d_{\mathcal X}(\bm X_t,\bm X_t^*) \right\|_{L_2} ~.
\]
Next define $d_{\mathcal Z}(\bm z_1,\bm z_2) = 12 C_1 ( |y_1 - y_2| + d_{\mathcal X}(\bm x_1,\bm x_2) )$ and note that $d_{\mathcal Z}$ is a distance, since it is a sum of distances rescaled by a positive constant. Moreover, it is straightforward to verify that the space $\mathcal{Z}$ equipped with the distance $d_{\mathcal Z}$ is also separable and complete which implies that $(\mathcal Z,d_{\mathcal Z})$ is Polish.
Since $(\mathcal{X}, d_{\mathcal{X}})$ is Polish by assumption, it suffices to prove
that $(\mathcal{Y}, |\cdot|)$ is Polish as well; in that case, the product space
$(\mathcal{Z}, d_{\mathcal{Z}})$ is automatically Polish. 
Recall that $\mathcal{Y} \subset \mathbb{R}$ is closed. 
Equipped with the metric induced by the Euclidean distance on $\mathbb{R}$, 
any subset of $\mathbb{R}$ is separable, and any closed subset is complete. 
It follows that $(\mathcal{Y}, |\cdot|)$ is a Polish space.
If we pick $\bm z=(0,\bm x)$ where $\bm x$ is is defined in \ref{cond:c2} we have that for any $s>0$ it holds that
\begin{align*}
	& \left\| d_{\mathcal Z}(\bm Z_t,\bm z) \right\|_{L_{2+s}} 
	= 12 C_1 \left\| | Y_t | + d_{\mathcal X}(\bm X_t,\bm x) \right\|_{L_{2+s}} \\ 
	& \quad \leq 12 C_1\left\| Y_t \|_{L_{2+s}} + 12 C_1\| d_{\mathcal X}(\bm X_t,\bm x) \right\|_{L_{2+s}} 
	\leq 12C_1 C_2^{(1)}\sqrt{s+2}(C_1+C_2) ~,
\end{align*}
where the last inequality follows from Proposition \ref{prop:subweibull:moments} ($C^{(1)}_2$ is defined in that proposition). Thus \ref{asm:coupling} is satisfied for 
$d_{\mathcal Z}(\bm z_1,\bm z_2) = 12 C_1 ( |y_1 - y_2| + d_{\mathcal X}(\bm x_1,\bm x_2) )$, $r=2$, $s=2$ and $C_{\mathcal Z}$ equal to the expression in the last display.

\paragraph{Applying Theorem \ref{thm:maintheorem}.}
We obtain the claim of the theorem by applying Theorem \ref{thm:maintheorem} and relying on the simplified version of the result in \eqref{eqn:simple}.
In particular, we obtain the claim by setting $\varepsilon = (1/\eta)\log(n)$. Using this choice of $\varepsilon$ we get
\[
	13 {T\over n} \exp( - \varepsilon ) = 13 {1\over n} 
	\text{ and }
	\beta^{1/4}\left( \left\lfloor{T \over n+1}\right\rfloor \right) \exp( \varepsilon )
	\leq n^{\zeta/4 -\zeta/(4\eta) + 1/\eta} = n^{-1/2}~,
\]
having used that $\eta = \frac{\zeta - 4}{\eta + 2}$. 
We remark that the conditions of the Theorem require $\zeta > 4$ so that $\eta > 0$ and consequently $\varepsilon > 0$.
\end{proof}

\begin{proof}[Proof of Corollary \ref{corollary:erm:neural_nets}]
	We verify that \ref{cond:c1} and \ref{cond:c2} hold for the single-layer neural network.
	The claim then follows from Proposition \ref{corollary:erm:reg}. 

	We introduce some additional notation and preliminary facts that will be used in the proof. 
	First, we note that $f_{\bm \theta}$ may be represented as
	\[ 
		f_{\bm \theta}(\bm X_t)=\sum_{k=1}^K\psi_k\sigma(\bm X_t'\bm w_k)=\left( \sigma( \bm X_t , \bm w_{1} , \ldots, \bm w_K )'{  \bm\psi  \over \| \bm\psi \|_2 } \right) \| \bm\psi \|_2
	\]
	where, $\sigma( \bm X_t , \bm w_{1} , \ldots, \bm w_K )=( \sigma( \bm X_t'\bm w_{1} ), \ldots, \sigma( \bm X_t'\bm w_{K} ) )'$ 
	and $\bm\psi=(\psi_1,\ldots,\psi_K)'$.
	Second, since $\sigma$ has bounded first sub-derivative, it follows that $\sigma$ is Lipschitz and we shall denote its Lipschitz constant by $L$.
	Third, since $\Theta$ is a compact set we have that we can find a positive constant $C_\Theta$ such that 
	$\sup_k|\psi_k| \leq C_\Theta$ and $\sup_k \|\bm{w}_k\|_2 \leq C_\Theta$.
	Fourth, since $\sigma$ is Lipschitz with Lipschitz constant $L$, $\sigma(0)=0$ and $\bm X_t$ is sub-Gaussian with $\| \bm X_t \|_{\psi_2} = \sigma_X$, 
	it follows that for all $k=1,\ldots,K$ we have that \( \|\sigma(\bm X_t'\bm w_k)\|_{\psi_2}\leq C' L C_\Theta \sigma_X  \) where $C'$ is a positive constant.
	To see this, note that if $\widetilde{\bm X}_t$ is an independent copy of $\bm X_t$ we have that, for any positive constant $c$,
	\begin{align*}
	& \mathbb E \exp {( \sigma(\bm X_t'\bm w_k) - \mathbb E \sigma(\bm X_t'\bm w_k) )^2\over c^2 } =\mathbb E \exp {( \sigma(\bm X_t'\bm w_k) - \mathbb E (\sigma(\widetilde{\bm X}_t'\bm w_k)\mid \bm X_t) )^2\over c^2 } \\
	&\quad\leq \mathbb E \left(\mathbb E\left( \left. \exp {( \sigma(\bm X_t'\bm w_k) - \sigma(\widetilde{\bm X}_t'\bm w_k) )^2\over c^2 }\right|\bm X_t\right)\right) 
	 \leq 
	\mathbb E\exp {L^2 (\bm X_t'\bm w_k - \widetilde{\bm X}_t'\bm w_k)^2 \over c^2 } \\
	&\quad =
	\mathbb E \exp {L^2 \|\bm w_k \|^2_2 (\bm X_t'\bm v - \widetilde{\bm X}_t'\bm v)^2 \over c^2 }  
	 \leq
	\mathbb E \exp {2 L^2 \|\bm w_k \|^2_2 [ (\bm X_t'\bm v)^2 + (\widetilde{\bm X}_t'\bm v)^2 ] \over c^2 }  \\
	& \quad =
	\left( \mathbb E \exp {2 L^2 \|\bm w_k \|^2_2 (\bm X_t'\bm v)^2 \over c^2 } \right)^2 
	 \leq
	\mathbb E \exp {4 L^2 \| \bm w_k \|^2_2 (\bm X_t'\bm v)^2 \over c^2 } ~,
\end{align*}
	where $\bm v = \bm w_k / \| \bm w_k \|_2$ and we remark that the first inequality follows from Jensen's inequality. 
	If we then set $c= 2 L \sigma_X \|\bm w_k \|_2$ we have that the expectation in the last expression is at most 2, 
	implying that $\|\sigma(\bm X_t'\bm w_k) - \mathbb E \sigma(\bm X_t'\bm w_k)\|_{\psi_2}=2L\sigma_X\|\bm w_k\|_2 \leq 2 L C_\Theta \sigma_X$. 
	Moreover, since $\sigma(0)=0$ it holds that 
	\begin{align*}
		& \| \mathbb E \sigma(\bm X_t'\bm w_k) \|_{\psi_2} \leq L \| \mathbb E |\bm X_t'\bm w_k| \|_{\psi_2} \leq {L \over \sqrt{\log(2)} } \sup_k \|\bm w_k\|_2  \| \bm X_t'\bm v \|_{L_1} 
		\leq {C^{(1)}_2  \over \sqrt{\log(2)} } L C_\Theta \sigma_X  ~,
	\end{align*}
	with $\bm v = \bm w_k / \| \bm w_k \|_2$, where the last inequality follows from Proposition \ref{prop:subweibull:moments} ($C^{(1)}_2$ is defined in that proposition).
	The result follows from the triangle inequality.

	\paragraph{Verifying \ref{cond:c1}} We verify that parts $(i)$ and $(ii)$ of condition \ref{cond:c1} hold for some positive constants $C_1$ and for the distance $d_\Theta(\bm \theta_1,\bm \theta_2) = \| \bm \theta_1 - \bm \theta_2 \|_2 $. \\
	$(i)$ We begin by noting that
	\begin{align*}
	& | f_{\bm \theta_{1}}(\bm X_{t}) - f_{\bm \theta_{2}}(\bm X_{t}) | = \left| \sum_{k=1}^K \psi_{1k} \sigma( \bm X_t'\bm w_{1k} ) - \psi_{2k} \sigma( \bm X ' \bm w_{2k} ) \right| \\
	& \quad \leq \left| \sum_{k=1}^K \psi_{1k} ( \sigma( \bm X_t'\bm w_{1k} ) - \sigma( \bm X_t'\bm w_{2k} ) ) \right| + \left|\sum_{k=1}^K \sigma( \bm X_t'\bm w_{2k} )( \psi_{1k} - \psi_{2k} ) \right| \\
	& \quad \leq L \sum_{k=1}^K |\psi_{1k}| |\bm X_t' (\bm w_{1k} - \bm w_{2k} ) | + \left|\sum_{k=1}^K \sigma( \bm X_t'\bm w_{2k} )( \psi_{1k} - \psi_{2k} ) \right| \\
	& \quad = L \sum_{k=1}^K |\psi_{1k}| \left|\bm X_t' {(\bm w_{1k} - \bm w_{2k} )\over \| \bm w_{1k} - \bm w_{2k} \|_2} \right| \| \bm w_{1k} - \bm w_{2k} \|_2
	+ \left| \sigma( \bm X_t, \bm w_{21}, \ldots , \bm w_{2K} )'{ ( \bm\psi_{1} - \bm\psi_{2} ) \over \| \bm\psi_{1} - \bm\psi_{2} \|_2 } \right| \| \bm\psi_{1} - \bm\psi_{2} \|_2 ~.
	\end{align*}
	Next we note that for any $k=1,\ldots,K$ it holds that
	\[
		\left\| \bm X_t' {(\bm w_{1k} - \bm w_{2k} )\over \| \bm w_{1k} - \bm w_{2k} \|_2} \right\|_{\psi_2} 
		\leq \left\| \bm X_t \right\|_{\psi_2} 
	= \sigma_X ~.
	\]
	Moreover, it holds that 
	\[
		\left\|  \sigma( \bm X_t, \bm w_{21},\ldots,\bm w_{2K} )'{ ( \bm\psi_{1} - \bm\psi_{2} ) \over \| \bm\psi_{1} - \bm\psi_{2} \|_2 } \right\|_{\psi_2} 
		\leq \left\|  \sigma( \bm X_t, \bm w_{21},\ldots,\bm w_{2K} ) \right\|_{\psi_2} \leq \sqrt{K} C' L C_\Theta \sigma_X ~. 
	\]
	Finally, combining these results,
	\begin{align*}
	& \| f_{\bm \theta_{1}}(\bm X_{t}) - f_{\bm \theta_{2}}(\bm X_{t}) \|_{\psi_2} \leq L C_\Theta \sigma_X \sum_{k=1}^K \| \bm w_{1k} - \bm w_{2k} \|_2 
		+ \sqrt{K} C' L C_\Theta \sigma_X    \| \bm\psi_{1} - \bm\psi_{2} \|_2 \\
	& \quad \leq (1+C') \sqrt K L C_\Theta \sigma_X  {K+1\over K+1} \left( \sum_{k=1}^K \| \bm w_{1k} - \bm w_{2k} \|_2 + \| \bm\psi_{1} - \bm\psi_{2} \|_2 \right) \\
	& \quad \leq (1+C') L C_\Theta \sigma_X  (K+1)^{3/2} \sqrt{  \| \bm w_{1} - \bm w_{2} \|^2_2 +\| \bm\psi_{1} - \bm\psi_{2} \|_2^2 \over K+1 } 
	= (1+C') L C_\Theta \sigma_X  (K+1)\|\bm\theta_1-\bm\theta_2\|_2 ~,
	\end{align*}
	where, the final inequality follows from Jensen's inequality. \\
	$(ii)$
Thus it follows that
\begin{align*}
	&  \left\|\sum_{k=1}^K\psi_k\sigma(\bm X_t'\bm w_k)\right\|_{\psi_2} \leq \|\bm\psi\|_2\left\|\sigma( \bm X_t , \bm w_{1},\ldots,\bm w_K )'{  \bm\psi  \over \| \bm\psi \|_2 } \right\|_{\psi_2} 
	\leq K C' L C_\Theta^2 \sigma_X  ~.
\end{align*}
 Additionally, by definition,  $\|Y_t\|_{\psi_2} = \sigma_Y$. These results confirm the second part of \ref{cond:c1}. 

\paragraph{Verifying \ref{cond:c2}}
	We verify that parts $(i)$, $(ii)$ and $(iii)$ of condition \ref{cond:c2} hold for some positive constant $C_2$, $\bm x=0$ and for the distance $d_{\mathcal X}(\bm x_1,\bm x_2) = (K L C_\Theta^2) \| \bm x_1-\bm x_2\|_\infty$. \\
	$(i)$ The fact that $(\mathcal X,d_\mathcal X)$ is Polish follows from the fact that $\R^d$ equipped with the sup-norm is Polish, and any closed subset of a Polish space is Polish as well.\\
	$(ii)$ We note that
	\begin{align*}
	&| f_{\bm \theta}(\bm x_1) - f_{\bm \theta}(\bm x_2) | 
	 \leq \sum_{k=1}^K |\psi_k ( \sigma( \bm x_2' \bm w_k ) - \sigma( \bm x_1' \bm w_k ) )| 
	\leq L \sum_{k=1}^K |\psi_k|  | (\bm x_2 - \bm x_1)' \bm w_k | \\
	& \leq L\sup_k|\psi_k|\sum_{k=1}^{K}|(\bm x_1-\bm x_2)'\bm w_k|= L\sup_k|\psi_k| \sum_{k=1}^{K} \|\bm w_k\|_1 \|\bm x_1-\bm x_2\|_\infty\\
	& \leq K L\sup_k|\psi_k| \sup_k  \|\bm w_k\|_1 \|\bm x_1-\bm x_2\|_\infty ~.
	\end{align*}
	$(iii)$ It follows from \citet{Vershynin}[Proposition 2.7.6] and by the fact that $\bm x=\bm 0$ that
	$\| \, \|\bm X-\bm x\|_\infty \, \|_{\psi_2} = \| \max_{i \in \{1,\ldots,p\} } |X_i| \|_{\psi_2}\leq C'' \sigma_X \sqrt{\log(d)}$ where $C''$ is defined in that proposition.
	Then by Proposition \ref{prop:subweibull:moments} we have that
	\[
		\| \, \| \bm X - \bm x \|_\infty \, \|_{L_{4}}\leq 2 C_\alpha^{(1)} C'' \sigma_X \sqrt{ \log(d) } ~,
	\]
	where $C_\alpha^{(1)}$ is a positive constant defined in Proposition \ref{prop:subweibull:moments} .

\paragraph{Applying Corollary \ref{corollary:erm:reg}.}

The claim of the corollary follows after noting that \cite[Corollary 4.2.11]{Vershynin} implies 
\begin{align*}
	& \gamma_1(\Theta) \leq (\log(2))^2 p \Delta(\Theta) \leq 2 K (\log(2))^2 d \Delta(\Theta) \\
	& \gamma_2(\Theta) \leq 2(\log(2))^{-1/2}\left(1-1/2^{1/2}\right)\sqrt{p}\Delta(\Theta) \leq 2(\log(2))^{-1/2}\left(1-1/2^{1/2}\right) \sqrt{2 K}\sqrt{d}\Delta(\Theta)~,
\end{align*}
where we recall that $\Delta(\Theta)$ is the diameter of the set $\Theta$.

\end{proof}

\subsection{Proofs for Section \ref{sec:mainproof}}

\begin{proof}[Proof of Proposition \ref{prop:coupling}]
	Part $(i)$ is an immediate consequence of \citep[Theorem 2.9]{MerlevedePeligrad2002} parts $(a)$, $(b)$, part $(ii)$ is a consequence of $(f)$ of the same theorem and
	part $(iii)$ is a consequence of \citep[Remark 2.5]{MerlevedePeligrad2002}.
\end{proof}

\begin{proof}[Proof of Proposition \ref{prop:subweibull:conc}]
	We begin by introducing a number of auxiliary quantities and basic facts.
	First, let $a_i=\|X_i\|_{\psi_\alpha}$ define $\widetilde{X}_i={X_i / a_i}$ and note that $\mathbb P(|\widetilde{X}_i |\geq \varepsilon)\leq 2\exp(-\varepsilon^\alpha)$.
	Second, let $\{ \epsilon_i \}_{i=1}^n$ denote a sequence of independent Rademacher random variables independent of $\{\widetilde{X}_i\}_{i=1}^n$ and note that $\epsilon_i \widetilde{X}_i$ is identically distributed as $\epsilon_i|\widetilde X_i|$. 
	Third, let $\zeta=(\log(2))^{1/\alpha}$ define $Y_i=(|\widetilde{X}_i|-\zeta )_+$ when $\alpha\geq 1$ and 
	$Y_i=(|\widetilde{X}_i|^\alpha-\log(2) )^{1/\alpha}_+$ when $\alpha<1$ and note that $\mathbb P(Y_i\geq\varepsilon)\leq \exp\left(-\varepsilon^\alpha\right)$ and that
	$|\widetilde{X}_i| \leq  2^{[(1 - \alpha)/\alpha] \wedge 0}( Y_i + \zeta )$ for all $\alpha>0$.\footnote{We remark that \citep[Theorem 3.1]{KuchibhotlaChakrabortty2022} define $Y_i=(|\widetilde{X}_i|-\zeta )_+$ for all $\alpha>0$. This however appears to be a typo.}
	In fact we have that when $\alpha\geq 1$ 
	\begin{align*}
		& \mathbb P( Y_i \geq \varepsilon ) =\mathbb P( |\widetilde{X}_i| \geq \varepsilon + \zeta ) 
		\leq 2 \exp( -(\varepsilon + \zeta )^\alpha ) \leq 2 \exp( -\varepsilon^\alpha - \zeta^\alpha ) = \exp( -\varepsilon^\alpha ) ~,
	\end{align*}
	and when $\alpha < 1$ we have that 
	\begin{align*}
		& \mathbb P( Y_i \geq \varepsilon ) = 
		\mathbb P( Y_i^\alpha \geq \varepsilon^\alpha ) = 
		\mathbb P( |\widetilde{X}_i|^\alpha \geq \varepsilon^\alpha + \log(2) ) =
		\mathbb P( |\widetilde{X}_i| \geq (\varepsilon^\alpha + \log(2) )^{1/\alpha} )\\
		& \quad \leq 2 \exp( -\varepsilon^\alpha - \log(2) ) 
		= \exp( -\varepsilon^\alpha ) ~.
	\end{align*}
	Moreover, when $\alpha\geq1$ the inequality $|\widetilde{X}_i| \leq Y_i + \zeta$ is immediate and when $\alpha <1$ we have that
	\[
		| \widetilde{X}_i | = ( | \widetilde{X}_i |^\alpha )^{1/\alpha} \leq ( Y_i^\alpha + \log(2) )^{1/\alpha} \leq 2^{(1 - \alpha)/\alpha} ( Y_i + \zeta ) ~.
	\]
	Fourth, note that the random variable $\epsilon_i Y_i$ is symmetric and satisfies $\mathbb P( |\epsilon_i Y_i| \geq\varepsilon)\leq \exp\left(-\varepsilon^\alpha\right)$.
	Then for any $p\geq 2$ we have
	\begin{align*}
		& \left\|\sum_{i=1}^nX_i\right\|_{L_p} \leq C_{X} \left\|\sum_{i=1}^n \widetilde{X}_i\right\|_{L_p} \leq 2 C_X \left\|\sum_{i=1}^n\epsilon_i \widetilde{X}_i\right\|_{L_p} =2 C_X \left\|\sum_{i=1}^n \epsilon_i|\widetilde{X}_i|\right\|_{L_p} \\ 
		& \quad = 2^{(1/\alpha) \wedge 1} C_X \left\|\sum_{i=1}^n \epsilon_i(Y_i+\zeta)\right\|_{L_p} \leq 2^{(1/\alpha) \wedge 1}C_X\left\|\sum_{i=1}^n \epsilon_iY_i\right\|_{L_p}+ 2^{(1/\alpha) \wedge 1}C_X \zeta\left\|\sum_{i=1}^n \epsilon_i\right\|_{L_p} \\
		& \quad \leq 2^{(1/\alpha) \wedge 1} C_X\left\|\sum_{i=1}^n \epsilon_iY_i\right\|_{L_p} + 2^{(1/\alpha) \wedge 1} C_X (\log(2))^{1/\alpha} \sqrt{ n } \sqrt p ~,
	\end{align*}
	where the second inequality follows from \citet[Proposition 6.3]{LedouxTalagrand2011} and, for $p>2$, the last inequality follows from \citet[Theorem 1.3.1]{delaPenaGine1999} (and for $p=2$ the inequality is trivial).	
	Next distinguish the cases $(i)$ $\alpha \leq 1$ and $(ii)$ $\alpha>1$.\\
	$(i)$ It follows from Proposition \ref{prop:latala} that for $p\geq 2$ there exists a positive constant $C$ (precisely defined in that proposition) that only depends on $\alpha$ such that	
	\[
		\left\|\sum_{i=1}^n \epsilon_i Y_i \right\|_{L_p}\leq C ( \sqrt{n}\sqrt p  + p^{1/\alpha} ) ~,
	\]
	where we have used the fact that $\epsilon_i Y_i$ is a symmetric random variable.
	Note that for $p=1$ we have that 
	\[
		\left\|\sum_{i=1}^n \epsilon_i Y_i \right\|_{L_1}
		\leq \left\|\sum_{i=1}^n \epsilon_i Y_i \right\|_{L_2}
		\leq C (\sqrt{n} \sqrt 2 +2^{1/\alpha}) ~.
	\]
	Thus, for $p\geq 1$ we have 
	\begin{align} 
		\left\|\sum_{i=1}^n \epsilon_i Y_i\right\|_{L_p} 
		\leq \max\{\sqrt 2,2^{1/\alpha}\} C ( \sqrt{n} \sqrt p +p^{1/\alpha} ) \label{eqn:subweibull:lesser1}~,
	\end{align} 
	and
	\begin{align*}
		\left\|\sum_{i=1}^nX_i\right\|_{L_p} 
		&\leq 2^{(1/\alpha) \wedge 1} [ C\max\{\sqrt 2,2^{1/\alpha}\}+ (\log(2))^{1/\alpha} ] C_X \sqrt n \sqrt p  +  2^{(1/\alpha) \wedge 1} C \max\{\sqrt 2,2^{1/\alpha}\}  C_X  p^{1/\alpha} ~.
	\end{align*}
	The claim then follows from Proposition \ref{prop:subweibull:GBOtailbound}. \\
	$(ii)$ Analogously to the case $(i)$, it follows from Proposition \ref{prop:latala} that for $p\geq 2$ there exists a positive constant $C$ (precisely defined in that proposition) that only depends on $\alpha$ such that 
	\begin{align*}
		\left\|\sum_{i=1}^n \epsilon_i Y_i \right\|_{L_p}\leq C \sqrt{p}\sqrt{n} + C  p^{1/\alpha}n^{(\alpha-1)/\alpha} ~.
	\end{align*}
	Note that for $p=1$,  we have that 
	\[
		\left\|\sum_{i=1}^n \epsilon_i Y_i \right\|_{L_1}\leq \left\|\sum_{i=1}^n \epsilon_i Y_i  \right\|_{L_2}\leq \max\{\sqrt 2,2^{1/\alpha}\} C ( \sqrt n+ n^{(\alpha-1)/\alpha} )~.
	\]
	Thus for $p\geq 1$ we have
	\begin{align}
		\left\|\sum_{i=1}^n \epsilon_i Y_i \right\|_{L_p}\leq \max\{\sqrt 2,2^{1/\alpha}\}C\left(\sqrt p\sqrt{n}+p^{1/\alpha}n^{(\alpha-1)/\alpha}\right)\label{eqn:subweibull:greater1} ~,
	\end{align}
	and
	\begin{align}
		\left\|\sum_{i=1}^n X_i \right\|_{L_p}\leq 2^{(1/\alpha) \wedge 1} \max\{\sqrt 2,2^{1/\alpha}\} C_X \left((C+(\log(2))^{1/\alpha}) \sqrt{n} \sqrt{p}+ 2^{(1/\alpha) \wedge 1}C n^{(\alpha-1)/\alpha} p^{1/\alpha} \right) ~.
	\end{align}
	The claim then follows from Proposition \ref{prop:subweibull:GBOtailbound}.
\end{proof}

\begin{proof}[Proof of Proposition \ref{prop:chaining:1}]
	Since $X_{\bm\theta}$ is a separable process there exists a countable dense subset $\tilde\Theta\subset\Theta$ such that 
	\[
		\sup_{\bm\theta\in\tilde\Theta}|X_{\bm\theta}-X_{\bm\theta_0}| =\sup_{\bm\theta\in\Theta}|X_{\bm\theta}-X_{\bm\theta_0}|~\text{a.s.}~,
	\]
	\citep[Ch.~11]{Boucheron}.
	Thus, the remainder of the proof consists of controlling $\sup_{\bm\theta\in\tilde\Theta}|X_{\bm\theta}-X_{\bm\theta_0}|$.
	Redefine $\tilde \Theta$ as $\tilde \Theta \cup \{ \bm \theta_0 \}$.
	Consider an admissible sequence of partitions $\{ \mathcal B_k \}_{k\geq 0}$ of $\tilde \Theta$ such that 
	for any $\bm \theta \in \tilde{\Theta}$ it holds that $\sum_{k\geq 0} 2^{k/\alpha} \Delta( B_k(\bm \theta) ) \leq 2 \gamma_\alpha(\tilde \Theta)$
	and an admissible sequence of partitions $\{ \mathcal C_k \}_{k\geq 0}$ of $\tilde \Theta$ such that for any $\bm \theta \in \tilde{\Theta}$ it holds that $\sum_{k\geq 0} 2^{k/2} \Delta( C_k(\bm \theta) ) \leq 2 \gamma_2(\tilde \Theta)$,
	where $B_k(\bm \theta)$ and $C_k(\bm \theta)$ are respectively the unique elements of $\mathcal B_k$ and $\mathcal C_k$ that contain $\bm \theta$.
	Define a new sequence of partitions $\{ \mathcal A_k \}_{k\geq 0}$ of $\tilde \Theta$ as follows. Set $\mathcal A_0 = \tilde \Theta$ and $\mathcal A_k$ as the partition 
	generated by $\mathcal B_{k-1}$ and $\mathcal C_{k-1}$, that is the partition that consists of the sets $B \cap C$ for $B \in \mathcal B_{k-1}$ and $C \in \mathcal C_{k-1}$.
	It is straightforward to verify that $\{ \mathcal A_k \}_{k\geq 0}$ is also an admissible sequence of partitions. 
	In fact, it holds that $|\mathcal A_k| \leq |\mathcal B_{k-1}||\mathcal C_{k-1}| \leq 2^{2^k}$.
	Moreover, for any $\bm \theta \in \tilde{\Theta}$ it holds that $\sum_{k\geq 0} 2^{k/\alpha} \Delta( A_k(\bm \theta) ) \leq 2 \gamma_\alpha(\tilde \Theta)$
	and $\sum_{k\geq 0} 2^{k/2} \Delta( A_k(\bm \theta) ) \leq 2 \gamma_2(\tilde \Theta)$ where $A_k(\bm \theta)$ is the unique element of $\mathcal A_k$ that contain $\bm \theta$.
	For each $k\geq 0$ consider the set $\tilde \Theta_k$ that intersects each element of $\mathcal A_k$ in exactly one point and set $\tilde \Theta_0 = \{ \bm \theta_0 \}$.
	For any \(\varepsilon > 0\), define the event \(\Omega(\varepsilon)\) as
\begin{align*}
	& \text{ for all } k \geq 1, \text{ for any } \bm\theta_1,\bm\theta_2\in \tilde\Theta_k, \\
	& \quad |X_{\bm \theta_1}-X_{\bm\theta_2}| \leq \left( a2^{(k+1)/2} \sqrt\varepsilon + b2^{(k+1)/\alpha}\varepsilon^{1/\alpha} \right)d_\Theta(\bm\theta_1,\bm\theta_2) 
	~.
\end{align*}
We proceed by bounding the probability of the event $\Omega^c(\varepsilon)$.
We begin by noting that for any $k$ and any $\bm\theta_1,\bm\theta_2\in \tilde\Theta_k$, we have
\begin{align*}
	& \mathbb P\left( |X_{\bm\theta_1}- X_{\bm\theta_2}|\geq \left(a2^{(k+1)/2}\sqrt \varepsilon+ b2^{(k+1)/\alpha}\varepsilon^{1/\alpha}\right) d_\Theta(\bm\theta_1,\bm\theta_2)\right) 
	\leq e\exp(-2^{k+1} \varepsilon) ~.
\end{align*}
By construction, $|\tilde\Theta_k|\leq 2^{2^k}$ implying that for $\varepsilon\geq 2$,
\begin{align*}
 	& \mathbb P( \Omega^c(\varepsilon) ) \leq e \sum_{k \geq 1} (|\tilde\Theta_k|)^2 \exp\left(-2^{k+1}\varepsilon\right) 
	< e \sum_{k \geq 1} 2^{2^{k+1}} \exp\left(-\varepsilon-2^{k}\varepsilon\right) \\
	& = e \exp(-\varepsilon) \sum_{k \geq 1} 2^{2^{k+1}} e^{-2^k \varepsilon } \leq 
	e \exp(-\varepsilon) \sum_{k \geq 1} \left({ 2 \over e }\right)^{2^{k+1}} 
	= 2 \exp(-\varepsilon) \sum_{k \geq 1} \left({ 2 \over e }\right)^{2^{k+1}-1} \leq 2\exp(-\varepsilon) ~.
\end{align*}
Let $\pi_k:\tilde\Theta\rightarrow  \tilde\Theta_k$ be the mapping such that $\pi_k(\bm\theta)=\arg\min_{s\in \tilde\Theta_k}d_\Theta(\bm\theta,s)$.
Then, assuming that $\Omega(\varepsilon)$ occurs, we apply the classic chaining argument which implies that
\begin{align}
	& \sup_{\bm\theta\in\tilde\Theta}|X_{\bm\theta}-X_{\bm\theta_0}|\leq \sup_{\bm\theta\in\tilde\Theta}\sum_{k\geq 1}|X_{\pi_k(\bm\theta)}-X_{\pi_{k-1}(\bm\theta)}| \nonumber \\
	& \quad \leq\sup_{\bm\theta\in\tilde\Theta}\sum_{k\geq 1}\left(a2^{(k+1)/2}\sqrt \varepsilon+b2^{(k+1)/\alpha}\varepsilon^{1/\alpha}\right)d_\Theta( \pi_k(\bm\theta), \pi_{k-1}(\bm\theta) ) \nonumber \\
	& \quad \leq a \sup_{\bm\theta\in\tilde\Theta}\left(\sum_{k\geq 1} 2^{(k+1)/2}\sqrt \varepsilon d_\Theta(\pi_{k-1}(\bm\theta),\bm\theta)+\sum_{k\geq 1} 2^{(k+1)/2}\sqrt \varepsilon d_\Theta(\pi_{k}(\bm\theta),\bm\theta)\right) \nonumber \\
	& \quad + b \sup_{\bm\theta\in\tilde\Theta} \left(\sum_{k\geq 1} 2^{(k+1)/\alpha}\varepsilon^{1/\alpha}d_\Theta(\pi_{k-1}(\bm\theta),\bm \theta)+\sum_{k\geq 1} 2^{(k+1)/\alpha}\varepsilon^{1/\alpha}d_\Theta(\pi_{k}(\bm\theta),\bm \theta)\right)\label{eqn:chaining:decomp1} ~,
\end{align}
where the last inequality follows from $d_\Theta(\pi_k(\bm\theta),\pi_{k-1}(\bm\theta))\leq d_\Theta(\bm\theta,\pi_{k-1}(\bm\theta))+d_\Theta(\pi_k(\bm\theta),\bm\theta)$. 
Next, we bound the two supremums on the right hand side of \eqref{eqn:chaining:decomp1}. 
First, we have that
\begin{align*}
	& \sup_{\bm\theta\in\tilde\Theta}\left(\sum_{k\geq 1} 2^{(k+1)/2}\sqrt \varepsilon d_\Theta(\pi_{k-1}(\bm\theta),\bm\theta)+\sum_{k\geq 1} 2^{(k+1)/2}\sqrt \varepsilon d_\Theta(\pi_{k}(\bm\theta),\bm\theta)\right)\\
	&\quad = \sup_{\bm\theta\in\tilde\Theta}\left(\sum_{k\geq 0} 2^{(k+2)/2}\sqrt \varepsilon d_\Theta(\pi_{k}(\bm\theta),\bm\theta)+\sum_{k\geq 1} 2^{(k+1)/2}\sqrt \varepsilon d_\Theta(\pi_{k}(\bm\theta),\bm\theta)\right)\\
	&\quad\leq 4 \sqrt \varepsilon \sup_{\bm\theta\in\tilde\Theta}\sum_{k\geq 0} 2^{k/2} \Delta( A_k(\bm \theta ) ) \leq 8 \gamma_2(\tilde\Theta)\sqrt \varepsilon \leq 8 \gamma_2(\Theta)\sqrt \varepsilon  ~,
\end{align*}
where the first step follows from reindexing and the last step follows from the fact that $\tilde\Theta\subset\Theta$.
Second, we have that (following analogous arguments)
\begin{align*}
	& \sup_{\bm\theta\in\tilde\Theta} \left(\sum_{k\geq 1} 2^{(k+1)/\alpha}\varepsilon^{1/\alpha}d_\Theta(\pi_{k-1}(\bm\theta),\bm \theta)+\sum_{k\geq 1} 2^{(k+1)/\alpha}\varepsilon^{1/\alpha}d_\Theta(\pi_{k}(\bm\theta),\bm \theta)\right)\\
	&\quad\leq \sup_{\bm\theta\in\tilde\Theta} \left(\sum_{k\geq 0} 2^{(k+2)/\alpha}\varepsilon^{1/\alpha}d_\Theta(\pi_{k}(\bm\theta),\bm \theta)+\sum_{k\geq 1} 2^{(k+1)/\alpha}\varepsilon^{1/\alpha}d_\Theta(\pi_{k}(\bm\theta),\bm \theta)\right)\\
	&\quad \leq 2^{(\alpha+2)/\alpha} \varepsilon^{1/\alpha}\sup_{\bm\theta\in\tilde\Theta}\sum_{k\geq 0} 2^{k/\alpha} \Delta( A_k(\bm \theta ) ) \leq 4^{(\alpha+1)/\alpha} \gamma_\alpha(\Theta)\varepsilon^{1/\alpha} ~.
\end{align*} 
Combining the previous results we obtain 
\begin{align*}
	\mathbb P\left(\sup_{\bm\theta\in\tilde\Theta}|X_{\bm\theta}-X_{\bm\theta_0}|
	\geq 8 a\gamma_2(\Theta)\sqrt\varepsilon+ 4^{(\alpha+1)/\alpha}b\gamma_\alpha(\Theta)\varepsilon^{1/\alpha}\right)\leq 2 \exp\left(-\varepsilon\right) ~,
\end{align*}
which implies the claim.
\end{proof}

rr

\begin{proof}[Proof of Proposition \ref{prop:gammafunctional:bound}]
The proof is based on a straightforward generalization of the arguments in \cite[pp.~22--23]{Talagrand2005}.
For the metric space $(\Theta,d_\theta)$ we define the $k$-th entropy number as $e_k(\Theta) = \inf_{\mathcal A_k} \sup_{\bm{\theta} \in \Theta} \Delta(A_k(\bm{\theta}))$. 
It is the smallest radius such that $\Theta$ can be covered by at most $2^{2^k}$ balls of that radius. 
Equivalently,
\[
e_k(\Theta) \;=\; \inf\Bigl\{\,\varepsilon>0:\ \mathcal N(\Theta,\varepsilon)\le 2^{2^k}\,\Bigr\} ~,
\]
where $\mathcal N(\Theta,\varepsilon)$ denotes the covering number of $\Theta$ at scale $\varepsilon$.
Next we note that for any $\varepsilon<e_k(\Theta)$ it holds that $\mathcal N(\Theta,\varepsilon)>2^{2^k}$, which in turn it implies that  $\mathcal N(\Theta,\varepsilon)\geq 2^{2^{k}}+1$. 
Thus, for any $k$ and $\varepsilon\in(e_{k+1},e_k)$, we have that 
\begin{align*}
	(\log(2^{2^{k}}+1))^{1/\alpha}(e_k(\Theta)-e_{k+1}(\Theta))&\leq \int_{e_{k+1}(\Theta)}^ {e_k(\Theta)}(\log(\mathcal N(\Theta,\varepsilon)))^{1/\alpha}\,d\varepsilon ~.
\end{align*}
Taking summation over $k\geq 0$ and using the facts that $(\log(2^{2^{k}}+1)) > 2^k \log(2)$ and $e_0(\Theta)\leq \Delta(\Theta)$ we obtain
\begin{align*}
	(\log(2))^{1/\alpha}\sum_{k\geq 0}2^{k/\alpha}(e_k(\Theta)-e_{k+1}(\Theta)) \leq  \int_{0}^ {\Delta(\Theta)}(\log(\mathcal N(\Theta,\varepsilon)))^{1/\alpha}\,d\varepsilon ~.
\end{align*}
Further, note that
\begin{align*}
	\sum_{k\geq 0}2^{k/\alpha}(e_k(\Theta)-e_{k+1}(\Theta))=\sum_{k\geq 0}2^{k/\alpha}e_k(\Theta)-\sum_{k\geq 1}2^{(k-1)/\alpha}e_{k}(\Theta)\geq \left(1-1/2^{1/\alpha}\right)\sum_{k\geq 0}2^{k/\alpha}e_k(\Theta)~.
\end{align*}
Combining the above two results we have 	
\begin{align*}
	& \gamma_\alpha(\Theta) =\inf_{\mathcal A_k}\sup_{\bm\theta\in\Theta}\sum_{k\geq 0}2^{k/\alpha}\Delta(A_k(\bm\theta))\\
	& \quad\leq\sum_{k\geq 0}2^{k/\alpha}e_k(\Theta)\leq (\log(2))^{1/\alpha}\left(1-1/2^{1/\alpha}\right)\int_{0}^ {\Delta(\Theta)}(\log(\mathcal N(\Theta,\varepsilon)))^{1/\alpha}\,d\varepsilon~,
\end{align*}
which establishes the claim.
\end{proof}

\subsection{Properties of sub-Weibull Random Variables}{\label{App:sub-Wei}}

This section collects several useful results on sub-Weibull random variables.
We begin by noting that a straightforward implication of the definition is that for a sub-Weibull($\alpha$) random variable $X$ of order $\alpha$ for some $\alpha>0$ it holds that
\begin{equation}\label{eq: tails subWei}
	\mathbb P( |X| \geq \varepsilon ) \leq 2 \exp\left( -\frac{\varepsilon^\alpha}{\| X \|_{\psi_\alpha}^\alpha} \right)
	\quad \text{for all } \varepsilon \geq 0 ~.
\end{equation}
In other words, sub-Weibull random variables have generalized exponential tails.

We remark that the notion of sub-Weibull random variable can be extended to sub-Weibull random vectors.
We say that the random vector $\bm X$ taking values in $\mathbb R^d$ is sub-Weibull($\alpha$) of order $\alpha$ for some $\alpha>0$ if
the one-dimensional marginals $\bm X'\bm v$ are sub-Weibull($\alpha$) for all $\bm v \in \mathbb R^d$. 
The sub-Weibull norm of $\bm X$ is defined as \( \| \bm X \|_{\psi_\alpha} = \sup_{ \bm v \in \mathcal S^{d-1} } \| \bm X' \bm v \|_{\psi_\alpha} \).

First, we show that any sub-Weibull($\alpha$) random variable of order $\alpha$ for some $\alpha>0$ belongs to $L^p$ for all $p \ge 1$. Moreover, the transition to the $L^p$ norm is explicit.

\begin{prop}\label{prop:subweibull:moments}
	Let $X$ be a sub-Weibull($\alpha$) random variable of order $\alpha$ for some $\alpha > 0$. 
	Then, for any integer $p \ge 1$ it holds that $\|X\|_{L_p}\leq C^{(1)}_\alpha \|X\|_{\psi_\alpha}p^{1/\alpha}$, 
	where $C^{(1)}_\alpha= 2\sqrt{2\pi}e^{\alpha/12 } e^{1/(2e)} \alpha^{ -(\alpha + 2)/(2\alpha) }$.
\end{prop}
\begin{proof} 
We have that
\begin{align*}
	& \mathbb{E}|X|^p = \int_{0}^\infty \mathbb{P}\left(|X|^p \geq u\right)\, du 
	= \int_{0}^\infty \mathbb{P}\left(|X| \geq u^{1/p}\right)\, du \\
	&\quad = \| X \|^p_{\psi_\alpha} \frac{p}{\alpha} \int_{0}^\infty \mathbb{P}\left(|X| \geq \left( \| X \|^p_{\psi_\alpha} t^{p/\alpha} \right)^{1/p} \right)\, t^{p/\alpha-1} dt ~,
\end{align*}
where the last equality follows from the change of variable \( u = \| X \|^p_{\psi_\alpha} t^{p/\alpha} \).
Using the tail bound for sub-Weibull($\alpha$) random variables given in \eqref{eq: tails subWei}, and recalling the definition of the Gamma function, we obtain
\begin{align*}
	& \mathbb{E}|X|^p
	\leq 2\|X\|_{\psi_\alpha}^p \frac{p}{\alpha} \int_{0}^\infty t^{p/\alpha - 1} \exp(-t)\, dt 
	= 2 \frac{p}{\alpha} \|X\|_{\psi_\alpha}^p \Gamma\left( \frac{p}{\alpha} \right) \\
	&\quad \leq 2 \sqrt{2\pi} \, e^{\alpha / (12p)} \sqrt{\frac{p}{\alpha}} \|X\|_{\psi_\alpha}^p \left( \frac{p}{e \alpha} \right)^{p/\alpha} ~,
\end{align*}
where the last inequality follows from Stirling’s approximation \citet[eq.~9.15]{feller1971}, which states that for all \( x > 0 \),
\[
\Gamma(x) < \sqrt{ \frac{2\pi}{x} } \left( \frac{x}{e} \right)^x e^{1/(12x)} ~.
\]
Taking the $p$-th root of the previous expression, we obtain
\begin{align*}
	& \left[ 2\sqrt{2\pi} e^{\alpha/(12p)} \sqrt{\frac{p}{\alpha}} \|X\|_{\psi_\alpha}^p \left( \frac{p}{e \alpha} \right)^{p/\alpha} \right]^{1/p} 
	= \left( 2\sqrt{2\pi} e^{\alpha/(12p)} \right)^{1/p} \, p^{1/(2p)} \, \alpha^{ -(\alpha + 2p)/(2p\alpha) } \, e^{-1/\alpha} \, \|X\|_{\psi_\alpha} \, p^{1/\alpha} \\
	& \quad \leq 2\sqrt{2\pi} \, e^{\alpha/12} \, e^{1/(2e)} \, \alpha^{ -(\alpha + 2)/(2\alpha) } \, \|X\|_{\psi_\alpha} \, p^{1/\alpha} ~,
\end{align*}
where the inequality follows from the fact that the function \( f(x) = x^{1/x} = e^{(\log x)/x} \) attains its maximum at \( x = e \), so that \( x^{1/x} \leq e^{1/e} \) for all \( x > 0 \). This completes the proof.
\end{proof}

Second, we show that the functional $\| \cdot \|_{\psi_\alpha}$ defines a norm when $\alpha \ge 1$ and a quasi-norm when $\alpha < 1$. 
Recall that a quasi-norm satisfies all the axioms of a norm, except that the triangle inequality holds only up to a multiplicative constant greater than one. 

\begin{prop}\label{prop:subweibull:sum}
	Let $X$ and $Y$ be sub-Weibull($\alpha$) random variables of order $\alpha$ for some $\alpha > 0$.
	Then, it holds that 
	\[
		\|X+Y\|_{\psi_\alpha}\leq C^{(2)}_\alpha(\|X\|_{\psi_\alpha}+\|Y\|_{\psi_\alpha}) ~,
	\]
	where $C^{(2)}_\alpha = 2^{1/\alpha}$ if $\alpha<1$ and $C^{(2)}_\alpha=1$ if $\alpha\geq1$.
\end{prop}
\begin{proof} 
When $\alpha < 1$, we exploit the fact that $|a + b|^\alpha \le |a|^\alpha + |b|^\alpha$ for any $a, b \ge 0$. We then have that
\begin{align*}
	&\mathbb{E}\left[\exp\left(\frac{|X+Y|^\alpha}{\left(2^{1/\alpha}(\|X\|_{\psi_\alpha} + \|Y\|_{\psi_\alpha})\right)^\alpha}\right)\right] \\
	&\quad \leq \mathbb{E}\left[\exp\left(\frac{|X|^\alpha}{\left(2(\|X\|_{\psi_\alpha} + \|Y\|_{\psi_\alpha})\right)^\alpha} + \frac{|Y|^\alpha}{\left(2(\|X\|_{\psi_\alpha} + \|Y\|_{\psi_\alpha})\right)^\alpha}\right)\right] \\
	&\quad \leq \mathbb{E}\left[\exp\left(\frac{|X|^\alpha}{2\|X\|_{\psi_\alpha}^\alpha} + \frac{|Y|^\alpha}{2\|Y\|_{\psi_\alpha}^\alpha}\right)\right] 
	= \mathbb{E}\left[\exp\left(\frac{|X|^\alpha}{2\|X\|_{\psi_\alpha}^\alpha}\right) \exp\left(\frac{|Y|^\alpha}{2\|Y\|_{\psi_\alpha}^\alpha}\right)\right] \\
	&\quad \leq \frac{1}{2} \left( \mathbb{E}\left[\exp\left(\frac{|X|^\alpha}{\|X\|_{\psi_\alpha}^\alpha}\right)\right] + \mathbb{E}\left[\exp\left(\frac{|Y|^\alpha}{\|Y\|_{\psi_\alpha}^\alpha}\right)\right] \right) ~,
\end{align*}
where the final inequality follows from Young's inequality, which states $ab \leq a^2/2 + b^2/2$.
Since $X$ and $Y$ are sub-Weibull($\alpha$), we know from \eqref{eq: def subWei} that each expectation on the right-hand side is bounded by $2$. We therefore deduce that $2^{1/\alpha}(\|X\|_{\psi_\alpha} + \|Y\|_{\psi_\alpha})$ belongs to the set on the right-hand side of \eqref{eq: def subWei}. By definition of the quasi-norm $\| X + Y \|_{\psi_\alpha}$ as the infimum of such quantities, it follows that
\( \| X + Y \|_{\psi_\alpha} \leq 2^{1/\alpha} (\|X\|_{\psi_\alpha} + \|Y\|_{\psi_\alpha}) \),
which concludes the proof in the case $\alpha < 1$. \\
When $\alpha \geq 1$, the sub-Weibull quasi-norm is in fact a norm, and the result follows directly from the triangle inequality. \citet{Vershynin} shows that the sub-Weibull quasi-norm is a true norm for 
$\alpha=1$ and $\alpha=2$, corresponding to sub-exponential and sub-Gaussian random variables.
\end{proof}

Third, from the bounds established in the two propositions above, we derive the following result concerning the centering of sub-Weibull random variables.

\begin{prop}\label{prop:subweibull:centering}
	Let $X$ be a sub-Weibull($\alpha$) random variable of order $\alpha$ with $\alpha > 0$. 
	Then, it holds that $\| X - \mathbb E[X] \|_{\psi_\alpha} < C_\alpha^{(3)} \|X\|_{\psi_\alpha}$, 
	where $C_\alpha^{(3)} = C_\alpha^{(2)}(1 + C^{(1)}_\alpha (\log 2)^{-\frac{1}{\alpha}})$. 
\end{prop}
\begin{proof}
Note that Proposition \ref{prop:subweibull:sum} yields the following bound
\[
	\|X - \mathbb{E}[X]\|_{\psi_\alpha} 
	\leq C^{(2)}_\alpha \left( \|X\|_{\psi_\alpha} + \| \mathbb{E}[X] \|_{\psi_\alpha} \right)
	\leq C^{(2)}_\alpha \left( \|X\|_{\psi_\alpha} + \| \mathbb{E}[|X|] \|_{\psi_\alpha} \right).
\]
We recall that the definition of the sub-Weibull random variable implies that
\[
\| \mathbb{E}[|X|] \|_{\psi_\alpha} = \inf \left\{ c > 0 : \exp\left( \left( \frac{\mathbb{E}[|X|]}{c} \right)^\alpha \right) \le 2 \right\}.
\]
Since $\mathbb{E}[|X|]$ is deterministic we may obtain a $c$ in the set on the right hand side of the display above by solving the equation $( \mathbb{E}[|X|] / c )^\alpha = \log 2$, which yields $ c = \mathbb{E}[|X|] (\log 2)^{-1/\alpha}$.
Substituting into the earlier inequality, we obtain
\[
\|X - \mathbb{E}[X]\|_{\psi_\alpha} \leq C^{(2)}_\alpha \left( \|X\|_{\psi_\alpha} + \log(2)^{-1/\alpha} \mathbb{E}[|X|] \right) 
\leq C^{(2)}_\alpha \left(1 + \log(2)^{-1/\alpha} C^{(1)}_\alpha \right) \|X\|_{\psi_\alpha},
\]
where the last inequality follows from Proposition \ref{prop:subweibull:moments}.  
This concludes the proof.
\end{proof} 

Fourth, we establish a bound on the $L^p$ norm of the sum of i.i.d.~symmetric random variables whose tails satisfy appropriate decay conditions, as specified below, akin sub-Weibull random variables.

\begin{prop}\label{prop:latala}
	Let $X_1, \ldots, X_n$ be a sequence of i.i.d.~symmetric random variables such that for any $\varepsilon>0$ it holds that $\mathbb P( | X_i | \geq \varepsilon ) \leq \exp( - \varepsilon^\alpha ) $~.

	Then, for any $p\geq 2$, we have that 
	$(i)$ if $\alpha < 1$ it holds that 
	\[
		\|X_1+X_2+\ldots+X_n\|_{L_p} \leq C^{(4)}_\alpha\left( p^{1/\alpha} +  \sqrt{p}\sqrt{n} \right) ~,
	\]
	where, $C^{(4)}_\alpha= 2e^3(2\pi)^{1/4}e^{1/24} (2e^{2/e}/\alpha)^{1/\alpha}$;
	and 
	$(ii)$ if $\alpha\geq 1$ it holds that 
	\[
		\|X_1+X_2+\ldots+X_n\|_{L_p}\leq C^{(4)}_\alpha ( p^{1/\alpha} n^{(\alpha-1)/\alpha}+ \sqrt{p} \sqrt{n} ) ~,
	\]
	where $C^{(4)}_\alpha=4e$.
\end{prop}

\begin{proof}
The proof relies on Theorem 2 of \citet{Latala1997}, which provides a bound on the \( L^p \) norm of the sum of symmetric random variables in terms of the Orlicz norm of the sequence. 
We begin by introducing the definition of the Orlicz norm for a sequence of random variables, and then proceed to show how it can be bounded as in the right-hand side of the statement.
For \( p > 0 \), the Orlicz norm of a sequence \( \{X_i\}_{i=1}^n \) is defined as
\[
	\vertiii{ \{ X_i \}_{i=1}^n }_{L_p} := \inf\left\{ t > 0 : \sum_{i=1}^n \log \mathbb E\left( \left| 1 + \frac{X_i}{t} \right|^p \right) \leq p \right\} ~.
\]
We now derive bounds for this quantity, considering separately the cases \( \alpha < 1 \) and \( \alpha \geq 1 \). \\
$(i)$ Case \(\alpha < 1\). 
We follow the argument in \citet[Example 3.3]{Latala1997}, which applies to random variables satisfying \( \mathbb P( |X_i| \geq \varepsilon ) \leq \exp(-N_\alpha(\varepsilon)) \) for $\varepsilon\geq 0$, 
where \( N_\alpha : \mathbb R_+ \rightarrow \mathbb R_+ \) is a concave function---a property satisfied when \( \alpha < 1 \).
Following \citet[Example 3.3]{Latala1997} the assumptions of the proposition imply that for any \( s > 0 \) and each \( i = 1, \dots, n \) we have
\[
	\log \left( \mathbb E \left( \left| 1 + \frac{sX_i}{e^2} \right|^p \right) \right)
	\leq p s^p \|X_i\|_{L_p}^p + p^2 s^2 \|X_i\|_{L_2}^2 ~.
\]
Setting \( s = e^2 / t \), we obtain
\[
	\vertiii{ \{ X_i \}_{i=1}^n }_{L_p}
	\leq \inf\left\{ t > 0 : \sum_{i=1}^n \left( \frac{e^{2p}}{t^p} \|X_i\|_{L_p}^p + \frac{p e^4}{t^2} \|X_i\|_{L_2}^2 \right) \leq 1 \right\} ~.
\]
By sub-additivity of the infimum, we can bound this by the sum of two terms:
\[
	\vertiii{ \{ X_i \}_{i=1}^n }_{L_p}
	\leq \inf\left\{ t > 0 : \frac{e^{2p}}{t^p} \sum_{i=1}^n \|X_i\|_{L_p}^p \leq 1 \right\}
	+ \inf\left\{ t > 0 : \frac{p e^4}{t^2}\sum_{i=1}^n \|X_i\|_{L_2}^2 \leq 1 \right\} ~.
\]
We now bound each of these terms separately.
For the first term, solving the equation \( \sum_{i=1}^n \frac{e^{2p}}{t^p} \|X_i\|_{L_p}^p = 1 \) and using the bound on the $L_p$ norm as in the proof of Proposition \ref{prop:subweibull:moments} gives
\[
	\inf\left\{ t > 0 : \frac{e^{2p}}{t^p} \sum_{i=1}^n  \|X_i\|_{L_p}^p \leq 1 \right\}
	\leq e^2 \left( \sum_{i=1}^n \|X_i\|_{L_p}^p \right)^{1/p} 
	\leq e^2 \left( n \frac{p}{\alpha} \Gamma\left( \frac{p}{\alpha} \right) \right)^{1/p} ~.
\]
Similarly, for the second term we have
\[
	\inf\left\{ t > 0 : \frac{p e^4}{t^2} \sum_{i=1}^n \|X_i\|_{L_2}^2 \leq 1 \right\}
	\leq e^2 \left( p \sum_{i=1}^n \|X_i\|_{L_2}^2 \right)^{1/2}
	\leq e^2 \left( p n \frac{2}{\alpha} \Gamma\left( \frac{2}{\alpha} \right) \right)^{1/2} ~.
\]
Combining both bounds and using Stirling’s approximation for the Gamma function, which states that 
\( \Gamma(x) < \sqrt{ (2\pi)/{x} } ( x/e )^x e^{1/(12x)} \) for all $x > 0$,
we obtain 
\begin{align*}
	& \vertiii{ \{ X_i \}_{i=1}^n }_{L_p} 
	\leq e^2 \left[ \left( n \frac{p}{\alpha} \Gamma\left( \frac{p}{\alpha} \right) \right)^{1/p}
	+ \left( p n \frac{2}{\alpha} \Gamma\left( \frac{2}{\alpha} \right) \right)^{1/2} \right] \\
	& \quad \leq e^2 \left\{ \left[ \sqrt{ \frac{2\pi p}{\alpha} } \left( \frac{p}{e\alpha} \right)^{p/\alpha} e^{\alpha/(12p)} \right]^{1/p} n^{1/p}
	+ \left[ \sqrt{ \frac{4\pi}{\alpha} } \left( \frac{2}{e\alpha} \right)^{2/\alpha} e^{\alpha/24} \right]^{1/2} \sqrt{p} \sqrt{n} \right\} \\
	& \quad \leq { e^2 (2\pi)^{1/4} e^{1/24} \over (e \alpha)^{1/\alpha}} \left\{ \left[ \left({ p \over \alpha}\right)^{1/(2p)} p^{1/\alpha} \right] n^{1/p}
	+ \left[ \left( \frac{2}{\alpha} \right)^{1/4} 2^{1/\alpha} \right] \sqrt{p} \sqrt{n} \right\} \\
	& \quad \leq { e^2 (2\pi)^{1/4} e^{1/24} \over (e\alpha)^{1/\alpha}} \left( \left\{ \left[ \left({ p \over \alpha}\right)^{\alpha/p} \right]^{1/(2\alpha)} p^{1/\alpha} \right\} n^{1/p}
	+ \left\{ \left[ \left( \frac{2}{\alpha} \right)^{\alpha/2} \right]^{1/(2\alpha)} 2^{1/\alpha} \right\} \sqrt{p} \sqrt{n} \right) \\
	& \quad \leq { e^2 (2\pi)^{1/4} e^{1/24} e^{1/(2 e \alpha)} \over (e\alpha)^{1/\alpha}} \left( p^{1/\alpha}  n^{1/p}
	+ 2^{1/\alpha}  \sqrt{p} \sqrt{n} \right) 
	\leq \frac{e^2 (2\pi)^{1/4} e^{1/24} 2^{1/\alpha}}{ \alpha^{1/\alpha} } \left( p^{1/\alpha} n^{1/p} + \sqrt{p} \sqrt{n} \right) ~,
\end{align*}
where in the third inequality we used that \( x^{1/x} \leq e^{1/e} \) for all \( x > 0 \), applied to \( x = {p/\alpha} \) and \( x = {2/\alpha} \), respectively.
Observe that, in the bound above, since \( p \ge 2 \) and \( \alpha < 1 \), the second term is smaller compared to the first. Therefore, we refine the bound further to eliminate the dependence on \( n^{1/p} \) in the first term, at the cost of enlarging the multiplicative constant. This step follows a trick presented in the proof of \citet[Corollary 1.2]{bogucki2015}.
Let \( \bm \iota \in \mathbb{R}^n \) be the vector with all components equal to one. By construction, we have \( \| \bm\iota \|_\infty = 1 \) and \( \| \bm\iota \|_p = n^{1/p} \) for any positive integer \( p \). Define
\begin{equation} \label{eq: choice C}
C(p, \alpha, n) := \left(p^{1/\alpha} \|\bm\iota\|_\infty + \sqrt{p} \|\bm\iota\|_2 \right)^{-1},
\end{equation}
which implies the bounds \( \| \bm\iota \|_\infty \le C(p, \alpha, n)^{-1} p^{-1/\alpha} \) and \( \| \bm\iota \|_2 \le C(p, \alpha, n)^{-1} p^{-1/2} \).
Now consider the following trick:
\[
n^{1/p} = \| \bm \iota \|_p = \left( \| \bm \iota \|_2^2 \| \bm \iota \|_\infty^{p-2} \right)^{1/p} 
\le C(p, \alpha, n)^{-1} \left( p^{-p/\alpha}  p^{(2-\alpha)/\alpha} \right)^{1/p}
\le C(p, \alpha, n)^{-1} \frac{e^{(2 - \alpha)/(e\alpha)}}{p^{1/\alpha}},
\]
where we have used that \( x^{1/x} \le e^{1/e} \) for \( x > 0 \), applied to \( x = p\).
Substituting the expression of \( C(p, \alpha, n) \) from \eqref{eq: choice C} into the above, we obtain:
\[
n^{1/p} \le \frac{e^{(2 - \alpha)/(e\alpha)}}{p^{1/\alpha}} \left( p^{1/\alpha} \| \bm\iota \|_\infty + \sqrt{p} \| \bm\iota \|_2 \right)
= \frac{e^{(2 - \alpha)/(e\alpha)}}{p^{1/\alpha}} \left( p^{1/\alpha} + \sqrt{p} \sqrt{n} \right).
\]
Inserting this bound into our previous estimate for the Orlicz norm yields
\begin{align*}
& \vertiii{ \{ X_i \}_{i=1}^n }_{L_p} \le \frac{e^2 2^{1/\alpha} (2\pi)^{1/4} e^{1/24} e^{(2 - \alpha)/(e\alpha)} }{ \alpha^{1/\alpha} } \left( p^{1/\alpha} + 2 \sqrt{p} \sqrt{n} \right) \\
& \quad \le 2 e^2 (2\pi)^{1/4} e^{1/24} \left( \frac{2 e^{2/e}}{\alpha} \right)^{1/\alpha} \left( p^{1/\alpha} + \sqrt{p} \sqrt{n} \right),
\end{align*}
which completes the proof for the case \( \alpha < 1 \), up to the application of \citet[Theorem 2]{Latala1997}.\\
$(ii)$ Case \(\alpha \geq 1\). We follow the argument in \citet[Example 3.2]{Latala1997}, which applies to random variables satisfying \( \mathbb P( |X_i| \geq \varepsilon ) \leq \exp(-N(\varepsilon)) \) for $\varepsilon\geq 0$, 
where \( N_\alpha : \mathbb R_+ \rightarrow \mathbb R_+ \) is a convex function---a property satisfied when \( \alpha \geq 1 \).
In this setting, for any \( s > 0 \) and \( i = 1, \ldots, n \), we have
\[
	\log \left( \mathbb{E}\left[ \left| 1 + \frac{s X_i}{4} \right|^p \right] \right) \leq 
	\begin{cases}
		N_\alpha^*(p|s|), & \text{if } p|s| \geq 2 \\
		p^2 s^2, & \text{if } p|s| < 2
	\end{cases}
	~
\]
where \( N_\alpha^*(y) = \sup_{x > 0} \{ yx - x^\alpha \} \) denotes the convex conjugate of the function \( N_\alpha(x) = x^\alpha \).  
%
Set \( s = 4/t \). Then, recalling the definition of the Orlicz norm and proceeding as in the proof of part $(i)$, we obtain
\begin{align*} 
	& \vertiii{ \{ X_i \}_{i=1}^n }_{L_p} 
	\le \inf \left\{ t > 0 : \sum_{i=1}^n N_\alpha^*\left( \frac{4p}{t} \right) \mathbf{1}_{\frac{4p}{t} \ge 2} + \sum_{i=1}^n p^2 \frac{16}{t^2} \mathbf{1}_{\frac{4p}{t} < 2} \le p \right\} \\
        &\quad \le \inf \left\{ t > 0 : \frac{n}{p} N_\alpha^*\left( \frac{4p}{t} \right) \le 1 \right\} 
    + \inf \left\{ t > 0 : \frac{16 n p}{t^2} \le 1 \right\} ~.
\end{align*}
Let us focus on the case \( \alpha = 1 \). It is straightforward to verify that
\[
	N_1^*(y) = \sup_{x > 0} \{ yx - x \} = 
	\begin{cases}
		0, & \text{if } y \leq 1 \\
		\infty, & \text{otherwise} ~,
	\end{cases}
\]
which leads to the bound
\[
	\inf \left\{ t > 0 : \frac{n}{p} N_1^*\left( \frac{4p}{t} \right) \le 1 \right\} = 
	\inf \left\{ t \geq 4p \right\} = 4p ~.
\]
Let us now focus on the case \( \alpha > 1 \). We now study the quantity
\[
	\inf \left\{ t > 0 : \frac{n}{p} N_\alpha^*\left( \frac{4p}{t} \right) \le 1 \right\}.
\]
Set \( y = {4p/t} \), and recall the definition of \( N_\alpha^*(y) \). We seek a pair \( (x, y) \) such that the function \( f(x, y) = xy - x^\alpha \) satisfies the constraints:
\[
    \begin{cases}
        f(x, y) \le \frac{p}{n} \\
        \partial_x f(x, y) = 0 ~.
    \end{cases}
\]
The second condition gives the optimizer \( x = \left( \frac{y}{\alpha} \right)^{\frac{1}{\alpha - 1}} \), which is positive for any $y>0$. Substituting this into \( f(x, y) \), we obtain:
\[
	xy - x^\alpha = y^{\frac{\alpha}{\alpha - 1}} \left( \left( \frac{1}{\alpha} \right)^{\frac{1}{\alpha - 1}} - \left( \frac{1}{\alpha} \right)^{\frac{\alpha}{\alpha - 1}} \right) 
	= y^{\frac{\alpha}{\alpha - 1}} \left( \frac{1}{\alpha} \right)^{\frac{\alpha}{\alpha - 1}} (\alpha - 1).
\]
This means that the first constraint in the system above is equivalent to
\[
	 \frac{n}{p} (\alpha - 1) \left( \frac{1}{\alpha} \right)^{\frac{\alpha}{\alpha - 1}} y^{\frac{\alpha}{\alpha - 1}} \le 1 ~.
\]
Substituting back \( y = \frac{4p}{t} \), we have
\[
	\inf \left\{ t > 0 : \frac{n}{p} N_\alpha^*\left( \frac{4p}{t} \right) \le 1 \right\}
	= \inf \left\{ t > 0 : \frac{n}{p} (\alpha - 1) \left( \frac{1}{\alpha} \right)^{\frac{\alpha}{\alpha - 1}} \left( \frac{4p}{t} \right)^{\frac{\alpha}{\alpha - 1}} \le 1 \right\} ~.
\]
This infimum is bounded above by the value of \( t \) solving the equality:
\[
	\frac{n}{p} (\alpha - 1) \left( \frac{1}{\alpha} \right)^{\frac{\alpha}{\alpha - 1}} \left( \frac{4p}{t} \right)^{\frac{\alpha}{\alpha - 1}} = 1 ~.
\]
Solving for \( t \), we get:
\[
	t = \frac{4p}{\alpha} \left( \frac{n(\alpha - 1)}{p} \right)^{\frac{\alpha - 1}{\alpha}} 
	\le 4 n^{\frac{\alpha - 1}{\alpha}} p^{\frac{1}{\alpha}} (\alpha - 1)^{ -\frac{1}{\alpha}} 
	\le 4 n^{\frac{\alpha - 1}{\alpha}} p^{\frac{1}{\alpha}} ~.
\]
Note that this expression matches the bound obtained earlier when \( \alpha = 1 \).
Furthermore, observe that
\[
	\inf \left\{ t > 0 : \frac{16np}{t^2} \le 1 \right\} \le 4 \sqrt{pn} ~.
\]
Combining the two bounds above yields
\[
	\vertiii{ \{ X_i \}_{i=1}^n }_{L_p} \le 4 p^{1/\alpha} n^{(\alpha - 1)/\alpha} + 4 \sqrt{p} \sqrt{n} ~.
\]
As in case $(i)$, the claim then follows by applying \citet[Theorem 2]{Latala1997}.
\end{proof}

Last, we conclude the appendix by stating and proving a proposition that goes in the opposite direction of the previous one. Instead of deriving moment bounds from tail behavior, we now use a bound on the
$L^p$ norm of a random variable to establish the exponential decay of its tails.

\begin{prop}\label{prop:subweibull:GBOtailbound}
	Let $X$ be a random variable such that, for some $\alpha >0$, it holds that $\|X\|_{L_p}\leq C_1\sqrt{p}+C_2p^{1/\alpha}$ for some positive constants $C_1$ and $C_2$ and any $p \geq 1$.

	Then it holds that
	\[
		\mathbb P\left(|X|\geq e C_1 \sqrt\varepsilon + e C_2 \varepsilon^{1/\alpha} \right)\leq e\exp(-\varepsilon) \text{ for all } \varepsilon \geq 0 ~.
	\]
\end{prop}
\begin{proof}
The proof of this result is based on a simplified version of the arguments used \citep[Proposition C.1 and A.3]{KuchibhotlaChakrabortty2022}.
Using Markov’s inequality and the bound on the \( L^p \) norm in the assumptions, we obtain for any \( \varepsilon \geq 1 \)
\begin{align*}
	& \mathbb{P}\left( |X| \geq eC_1 \sqrt{\varepsilon} + eC_2 \varepsilon^{1/\alpha} \right) 
	= \mathbb{P}\left( |X|^\varepsilon \geq \left( eC_1 \sqrt{\varepsilon} + eC_2 \varepsilon^{1/\alpha} \right)^\varepsilon \right) \\
	& \quad \leq \frac{\mathbb{E}|X|^\varepsilon}{\left( eC_1 \sqrt{\varepsilon} + eC_2 \varepsilon^{1/\alpha} \right)^\varepsilon} 
	\leq \frac{ \left( C_1 \sqrt{\varepsilon} + C_2 \varepsilon^{1/\alpha} \right)^\varepsilon }{ \left( eC_1 \sqrt{\varepsilon} + eC_2 \varepsilon^{1/\alpha} \right)^\varepsilon } 
	= \exp(-\varepsilon) ~.
\end{align*}
The claim of the proposition then follows.
\end{proof}

\singlespacing

\bibliographystyle{abbrvnat}
\setcitestyle{authoryear,open={((},close={))}}
\bibliography{references}

@book{Hastieetal:2009,
	author = {Trevor Hastie and Robert Tibshirani and Jerome Friedman},
	title = {The Elements of Statistical Learning: Data Mining, Inference and Prediction},
	publisher = springer,
	edition = "2nd",
	year = 2009
}

@book{Doukhan:1994,
	title = {{Mixing}},
	subtitle = {Properties and Examples},
	author = {Doukhan, Paul},
	year = {1994},
	publisher = {Springer-Verlag},
	address = {New York}
}

@article{Latala1997,
author = {Latala, Rafal},
title = "{Estimation of Moments of Sums of Independent Real Random Variables}",
journal = {The Annals of Probability},
volume = {25},
number = {3},
pages = {1502-1513},
year = {1997}
}

@book{delaPenaGine1999,
author = {de~la~Pe\~na, V. and Gin\`e, E.},
address = {New York, New York State},
booktitle = {Decoupling: from dependence to independence},
edition = {1st ed. 1999.},
isbn = {1-4612-0537-9},
language = {eng},
publisher = {Springer},
series = {Probability and Its Applications},
title = {Decoupling : from dependence to independence },
year = {1999}
}

@book{Devroye:Gyorfi:Lugosi:1996,
title = {{A Probabililstic Theory of Pattern Recognition}},
author = {Devroye, L. and Gy\"orfi, L. and Lugosi, G.},
year = {1996},
publisher = {Springer},
address = {New York}
}

@article{WongTewari:2020,
author = {Kam Chung Wong and Zifan Li and Ambuj Tewari},
title = {{Lasso guarantees for $\beta$-mixing heavy-tailed time series}},
volume = {48},
journal = {The Annals of Statistics},
number = {2},
publisher = {Institute of Mathematical Statistics},
pages = {1124 -- 1142},
year = {2020},
doi = {10.1214/19-AOS1840},
URL = {https://doi.org/10.1214/19-AOS1840}
}

@article{Brownlees:Gudmundsson:2021,
title={Performance of Empirical Risk Minimization for Linear Regression with Dependent Data}, 
author={Brownlees, C. and Gudmundsson, G. S.},
journal = {Econometric Theory},
volume = {forthcoming},
pages = {},
year={2025},
}

@article{Brownlees:LlorensTerrazas:2021,
title={Empirical Risk Minimization for Time Series: Nonparametric Performance Bounds for Prediction}, 
author={Brownlees, C. and Llorens-Terrazas, J.},
journal = {Journal of Econometrics},
volume = {forthcoming},
pages = {},
year={2025},
}

@book{MerlevedePeligrad2002,
author="Merlev{\`e}de, Florence and Peligrad, Magda",
editor="Dehling, Herold
and Mikosch, Thomas
and S{\o}rensen, Michael",
title="On the Coupling of Dependent Random Variables and Applications",
bookTitle="Empirical Process Techniques for Dependent Data",
year="2002",
publisher="Birkh{\"a}user Boston",
address="Boston, MA",
pages="171--193"
}

@article{KuchibhotlaChakrabortty2022,
author = {Kuchibhotla, Arun Kumar and Chakrabortty, Abhishek},
title = "{Moving beyond sub-Gaussianity in high-dimensional statistics: applications in covariance estimation and linear regression}",
journal = {Information and Inference: A Journal of the IMA},
volume = {11},
number = {4},
pages = {1389-1456},
year = {2022}
}

@book{Vershynin , 
place={Cambridge}, 
edition={2}, 
series={Cambridge Series in Statistical and Probabilistic Mathematics}, title={High-Dimensional Probability: An Introduction with Applications in Data Science}, publisher={Cambridge University Press}, 
author={Vershynin, Roman}, 
year={2026}, 
collection={Cambridge Series in Statistical and Probabilistic Mathematics}}

@book{Kosorok,
    author    = "M. R. Kosorok",
    title     = "Introduction to Empirical Processes and Semiparametric Inference",
    year      = "2008",
    publisher = "Springer",
    address   = "Berlin"
}

@book{Vandervaart,
    author    = "Aad W. Vaart, Jon A. Wellner",
    title     = "Weak Convergence and Empirical Processes:With Applications to Statistics",
    year      = "2012",
    publisher = "Springer",
    address   = "New York"
}

@book{LedouxTalagrand2011,
    author="Michel Ledoux and Michel Talagrand",
    title="Probability in Banach Spaces: Isoperimetry and Processes",
    volume = {23},
    year="1991",
    publisher="Springer-Verlag",
    address="Berlin"
}

@book{Boucheron,
  title="Concentration Inequalities: A Nonasymptotic Theory of Independence",
  author="Boucheron, S. and Lugosi, G. and Massart, P.",
  year="2013",
  publisher="Oxford University Press"
}

@article{bogucki2015,
  title={{Suprema of canonical Weibull processes}},
  author={Bogucki, Robert},
  journal={Statistics \& Probability Letters},
  volume={107},
  pages={253--263},
  year={2015},
  publisher={Elsevier}
}

@article{Jiang:Tanner:2010,
	author = {Jiang, Wenxin and Tanner, Martin},
	journal = {Econometric Theory},
	title = {{Risk Minimization for Time Series Binary Choice with Variable Selection}},
	year = {2010},
	volume = {26},
	pages = {1437-1452}
}

@book{pollard1984convergence,
  title={Convergence of Stochastic Processes},
  author={Pollard, D.},
  isbn={9780387909905},
  lccn={lc84001401},
  series={Clinical Perspectives in Obstetrics and Gynecology},
  url={https://books.google.es/books?id=B2vgGMa9vd4C},
  year={1984},
  publisher={Springer New York}
}

@inbook{Wainwright_2019, 
place={Cambridge}, 
series={Cambridge Series in Statistical and Probabilistic Mathematics}, 
title={Metric entropy and its uses}, 
booktitle={High-Dimensional Statistics: A Non-Asymptotic Viewpoint}, 
publisher={Cambridge University Press}, 
author={Wainwright, Martin J.}, 
year={2019}, 
pages={121–158}, 
collection={Cambridge Series in Statistical and Probabilistic Mathematics}}

@book{vandegeer2000,
  author    = {van de Geer, Sara A.},
  title     = {Empirical Processes in M‑Estimation},
  year      = {2000},
  publisher = {Cambridge University Press}
}

@book{feller1971,
  title     = {An Introduction to Probability Theory and Its Applications, Volume I},
  author    = {Feller, William},
  year      = {1971},
  edition   = {3rd},
  publisher = {John Wiley \& Sons},
  address   = {New York}
}

@book{Talagrand2005,
  author    = {Michel Talagrand},
  title     = {The Generic Chaining: Upper and Lower Bounds of Stochastic Processes},
  publisher = {Springer Berlin, Heidelberg},
  year      = {2005},
  series    = {Springer Monographs in Mathematics},
  edition   = {1},
  pages     = {VIII + 222},
  doi       = {10.1007/3-540-27499-5},
  isbn      = {978-3-540-24518-6},
  eisbn     = {978-3-540-27499-5},
  issn      = {1439-7382},
  eissn     = {2196-9922},
  url       = {https://doi.org/10.1007/3-540-27499-5},
}

\end{document}